\documentclass[11pt]{article}
\usepackage[latin9]{inputenc}
\usepackage[a4paper,margin=1in]{geometry}
\usepackage{verbatim}
\usepackage{float}
\usepackage{amsthm}
\usepackage{amsmath}
\usepackage{amssymb}
\usepackage{graphicx}

\makeatletter

\floatstyle{ruled}
\newfloat{algorithm}{tbp}{loa}
\providecommand{\algorithmname}{Algorithm}
\floatname{algorithm}{\protect\algorithmname}

\theoremstyle{plain}
\newtheorem{thm}{\protect\theoremname}
  \theoremstyle{definition}
  \newtheorem{problem}{\protect\problemname}
  \theoremstyle{definition}
  \newtheorem{defn}{\protect\definitionname}
  \theoremstyle{plain}
  \newtheorem{lem}{\protect\lemmaname}
\newenvironment{lyxlist}[1]
{\begin{list}{}
{\settowidth{\labelwidth}{#1}
 \setlength{\leftmargin}{\labelwidth}
 \addtolength{\leftmargin}{\labelsep}
 }}
{\end{list}}

\@ifundefined{date}{}{\date{}}
\usepackage{authblk}
\usepackage{algpseudocode}
\usepackage{algorithm}\usepackage{enumitem}

\usepackage{amsthm}
\usepackage{amsfonts}\usepackage{xspace}
\usepackage{upgreek}
\usepackage{braket}
\usepackage{enumitem}

\usepackage[usenames,dvipsnames]{xcolor}

\usepackage{amsfonts}

\usepackage[backend=bibtex,firstinits,style=numeric-comp,maxnames=99]{biblatex}

\DeclareFieldFormat{titlecase}{\MakeTitleCase{#1}}

\newrobustcmd{\MakeTitleCase}[1]{%
  \ifthenelse{\ifcurrentfield{booktitle}\OR\ifcurrentfield{booksubtitle}%
    \OR\ifcurrentfield{maintitle}\OR\ifcurrentfield{mainsubtitle}%
    \OR\ifcurrentfield{journaltitle}\OR\ifcurrentfield{journalsubtitle}%
    \OR\ifcurrentfield{issuetitle}\OR\ifcurrentfield{issuesubtitle}%
    \OR\ifentrytype{REMOVEIFEXCLUDEBOOKbook}\OR\ifentrytype{mvbook}\OR\ifentrytype{bookinbook}%
    \OR\ifentrytype{booklet}\OR\ifentrytype{suppbook}%
    \OR\ifentrytype{collection}\OR\ifentrytype{mvcollection}%
    \OR\ifentrytype{suppcollection}\OR\ifentrytype{manual}%
    \OR\ifentrytype{periodical}\OR\ifentrytype{suppperiodical}%
    \OR\ifentrytype{proceedings}\OR\ifentrytype{mvproceedings}%
    \OR\ifentrytype{reference}\OR\ifentrytype{mvreference}%
    \OR\ifentrytype{report}\OR\ifentrytype{thesis}}
    {#1}
    {\MakeSentenceCase{#1}}}

\bibliography{longnames,bib-latest,bristol-ESA}

\renewcommand{\geq}{\geqslant}
\renewcommand{\leq}{\leqslant}
\renewcommand{\epsilon}{\varepsilon}
\renewcommand{\Delta}{\Updelta}

\newcommand{\pav}{\ensuremath{\mathcal{A}_v}}
\newcommand{\column}{\ensuremath{\textup{column}}}
\newcommand{\row}{\ensuremath{\textup{row}}}
\newcommand{\toprow}{\ensuremath{\textup{toprow}}}
\newcommand{\bottomrow}{\ensuremath{\textup{bottomrow}}}
\newcommand{\dist}{\ensuremath{\textup{columndist}}}
\newcommand{\cost}{\ensuremath{\textup{cost}}}

\newcommand{\calA}{\ensuremath{\mathcal{A}}}

\newcommand{\calC}{\ensuremath{\mathcal{C}}}
\newcommand{\calD}{\operatorname{ED}}

\newcommand{\calT}{\ensuremath{\mathcal{T}}}

\newcommand{\LCS}{\ensuremath{\textup{LCS}}}
\newcommand{\len}{\ensuremath{\textup{length}}}

\newcommand{\expected}[1]{\ensuremath{\mathbb{E}[#1]}}
\newcommand{\expecteddisplay}[1]{\ensuremath{\mathbb{E}\left[#1\right]}}

\newcommand{\ED}[1]{\ensuremath{\textup{Edit}(#1)}}

\newcommand{\LCSt}{\ensuremath{\textup{LCS}(F,S_t)}}
\newcommand{\HAMt}{\ensuremath{\textup{Ham}(F,S_t)}}

\newcommand{\Xv}{\ensuremath{\mathcal{S}_v}}
\newcommand{\Xvknown}{\ensuremath{\widetilde{\mathcal{S}}_v}}
\newcommand{\Xvfix}{\ensuremath{\widetilde{s}_v}}
\newcommand{\Yv}{\ensuremath{Y_v}}

\newcommand{\vitset}{\ensuremath{\mathcal{I}_v}}
\newcommand{\vit}{\ensuremath{I_v}}
\newcommand{\vtime}{\ensuremath{R_v}}
\newcommand{\vcap}{\ensuremath{r_v}}
\newcommand{\venc}{\ensuremath{Z_v}}

\newcommand{\pad}{\ensuremath{\rho}} 
\newcommand{\pchar}{\ensuremath{\diamond}} 
\newcommand{\bchar}{\ensuremath{\texttt{\textup{h}}}} 
\newcommand{\schar}{\ensuremath{\texttt{\textup{t}}}} 
\newcommand{\fchar}{\ensuremath{\star}} 
\newcommand{\PA}{well-aligned\xspace}

\newcommand{\pred}{\rho}
\newcommand{\mask}{\textup{predecessor}}
\newcommand{\Dinter}{\ensuremath{\widetilde{D}}}
\newcommand{\progression}{\textup{block}}
\newcommand{\first}{\ensuremath{\alpha}}

\newcommand{\val}[1]{\text{val}(#1)}

\newcommand{\D}{{\ensuremath{D}}} 

\newcommand{\Patrascu}{P{\v a}tra{\c s}cu\xspace}

\makeatother

  \providecommand{\definitionname}{Definition}
  \providecommand{\lemmaname}{Lemma}
  \providecommand{\problemname}{Problem}
\providecommand{\theoremname}{Theorem}

\begin{document}


\title{Cell-Probe Bounds for Online Edit Distance and\\
 Other Pattern Matching Problems}

\author{Raphaël Clifford \quad{}~ Markus Jalsenius \quad{}~ Benjamin
Sach\\
 {\small{}Department of Computer Science}\\
{\small{}University of Bristol}\\
{\small{}Bristol, UK}}
\maketitle
\begin{abstract}
We give cell-probe bounds for the computation of edit distance, Hamming
distance, convolution and longest common subsequence in a stream.
In this model, a fixed string of $n$ symbols is given and one $\delta$-bit
symbol arrives at a time in a stream. After each symbol arrives, the
distance between the fixed string and a suffix of most recent symbols
of the stream is reported. The cell-probe model is perhaps the strongest
model of computation for showing data structure lower bounds, subsuming
in particular the popular word-RAM model. 
\begin{itemize}
\item We first give an $\Omega\big((\delta\log n)/(w+\log\log n)\big)$
lower bound for the time to give each output for both online Hamming distance
and convolution, where $w$ is the word size. This bound relies on
a new encoding scheme and for the first time holds even when $w$
is as small as a single bit. {\small \par}
\item We then consider the online edit distance and longest common subsequence
problems in the bit-probe model ($w=1$) with a constant sized input
alphabet. We give a lower bound of $\Omega(\sqrt{\log{n}}/(\log{\log{n}})^{3/2
})$
 which applies for both problems. This second set of results
relies both on our new encoding scheme as well as a carefully constructed
hard distribution. {\small \par}
\item Finally, for the online edit distance problem we show that there is an $
O((\log^{2}{n})/w)$
upper bound in the cell-probe model. This bound gives a contrast to our new lower bound and also establishes an exponential
gap between the known cell-probe and RAM model complexities. {\small \par}
\end{itemize}
\end{abstract}

\section{Introduction}

The search for lower bounds in general random-access models of computation
provides some of the most important and challenging problems within
computer science. In the offline setting where all the data representing
the problem are given at once, non-trivial, unconditional, time lower
bounds still appear beyond our reach. One area where there has however
been success in proving time lower bounds is in the field of dynamic
data structure problems (see for example~\cite{FS1989:chronogram,PD2006:Low-Bounds,
Pat2008:Thesis,Larsen:2012,Larsen:2012:focs}
and references therein) and online or streaming problems~\cite{CJ:2011,CJS:2013}.

We consider streaming problems where there is a given fixed array
of $n$ values and a separate stream of values that arrive one at
a time. After each value arrives in the stream, a function of the
fixed array and the latest values of the stream is computed and reported.
We consider several fundamental measures: edit distance, Hamming distance,
inner product/convolution and longest common subsequence (LCS). Efficiently
finding patterns in massive and streaming data is a topic of considerable
interest because of its wide range of practical applications.

For each problem we give new cell-probe lower bounds. For the edit
distance problem we also provide a cell-probe upper bound that is
exponentially faster than was previously known. The cell-probe model
is perhaps the strongest model of computation for showing data structure
lower bounds, subsuming in particular the popular word-RAM model.

\paragraph{Hamming distance and convolution}
Our first set of results concern cell-probe lower bounds for Hamming distance and convolution.
For brevity, we write $[q]$ to denote the set $\{0,\dots,q-1\}$,
where $q$ is a positive integer.
\begin{problem}
[\textbf{Online Hamming distance and convolution}]\label{problem:hamming}
For a fixed array $F\in[q]^{n}$ of length $n$, we consider a stream
in which symbols from $[q]$ arrive one at a time. In the online \emph{Hamming~distance~problem},
for each arriving symbol, before the next value
arrives, we output the Hamming distance between $F$ and the array consisting of the latest $n$ values of the stream. In the \emph{convolution~problem}, the output is instead the inner product between $F$ and
the latest $n$ values in the stream.
\end{problem}


We show a lower bound for the expected amortised time per output of
any randomised algorithm that solves either the Hamming distance or convolution problem.
Throughout the paper we let $w$
denote the word size and $\delta=\lfloor\log_{2}{q}\rfloor$.

\begin{thm}
\label{thm:conham} In the cell-probe model with any word size $w$
and positive integer $q$, the expected amortised time per output
of any randomised algorithm that solves either the Hamming distance
problem or the convolution problem is
\[
\Omega{\left(\frac{\delta\log n}{w+\log\log n}\right)}.
\]
This lower bound also holds even when outputs are reported modulo~$q$.
\end{thm}
These lower bounds can be compared with the known $O((\delta/w)\log{n})$
time upper bounds for both the online Hamming distance and convolution in
the cell-probe model~\cite{CEPP:2011,CJS:2013}. To obtain these results we first provide a new method that gives a
straightforward and clean unifying framework for streaming lower bounds
in the cell-probe model with small word sizes. This itself marks a
methodological advance in the development of lower bounds for streaming
problems.

Cell-probe lower bounds of $\Omega((\delta/w)\log{n})$ time per output
for the convolution problem~\cite{CJ:2011} and Hamming distance
problem~\cite{CJS:2013} had previously been shown only when the
word size $w\in\Omega(\log{n})$. These bounds were therefore meaningful
only when $\delta$, the number of bits needed to represent a symbol,
was sufficiently large, typically $\Theta(\log{n})$. From a practical
point of view it is arguable that such very large input alphabets
are rarely seen. Our improved lower bound gives a smooth trade-off
that holds for word and symbol sizes as small as a single bit. In
the perhaps most interesting case where $\delta=w$ we therefore show
an $\Omega(\log{n})$ lower bound for $w\in\Omega(\log{\log{n}})$.
In the bit-probe model~(i.e. $w=1$), our results imply an $\Omega(\log^{2}{n}/\log{\log{n}})$
lower bound when $\delta=\Omega(\log{n})$. We are also able to obtain
interesting lower bounds for smaller values of $\delta$, in particular
an $\Omega(\log{n}/\log{\log{n}})$ lower bound for binary inputs
($\delta=1$) in the bit-probe model. The bit-probe model is
considered theoretically appealing due to its machine independence
and overall cleanliness~\cite{M:1993,PP:2007:dynamicbitprobe}.

In the unit-cost word-RAM model with $w, \delta \in\Theta(\log n)$, the Hamming distance problem
can be solved in $O(\sqrt{n\log{n}})$ time per arriving symbol, and
the convolution problem can be solved in $O(\log^{2}n)$ time per
arriving symbol~\cite{CEPP:2011}. The lower bounds we give therefore might appear distant from the known RAM model upper bounds in the case of Hamming distance. However, an online-to-offline reduction of~\cite{CEPP:2011} also tells us that any improvement to the online lower bound in the RAM model would imply a new super-linear and offline RAM model lower bound. As no such lower bound is known for any problem in NP this seems a not inconsiderable barrier.

\paragraph{Edit distance and LCS} Our next set of results concern edit distance and longest common subsequence
(LCS), for which no previous, non-trivial cell-probe bounds were known.
For the edit distance and LCS problems we define $S(k)$ to be the
$k$-length string representing the most recent $k$ symbols in the
stream and $\ED{A,B}$ to be the minimal number of single symbol edit
operations (replace, delete and insert) required to transform string
$A$ into string $B$. Analogously, $\LCS(A,B)$ is defined to be the length of an LCS of $A$ and $B$.  

In our online setting, the natural way of defining
edit distance is by minimising the distance over all suffixes of the
stream. For the related LCS problem we take a slightly different
approach to avoid the LCS rapidly converging to the entire fixed string
$F$. As a result we consider a fixed length sliding window of the
stream instead, as in the Hamming distance problem.

\begin{problem}
[\textbf{Online edit distance and LCS}]\label{problem:editdistance}
For a fixed string $F\in[q]^{n}$ of length~$n$, we consider a stream
in which symbols from $[q]$ arrive one at a time. In the \emph{edit~distance~problem},
for each arriving symbol we output $\min_{k}\ED{F,S(k)}$.
In the \emph{longest~common~subsequence (LCS) problem}, for each arriving symbol
we output $\LCS(F,S(n))$.
\end{problem}

We show a lower bound for the expected amortised time per output of
any randomised algorithm that solves either the edit distance problem
or the LCS problem.

\begin{thm}
\label{thm:lcs} In the cell-probe model, where the word size $w=1$
and the alphabet size is at least~4, the expected amortised time
per output of any randomised algorithm that solves either the edit
distance problem or the LCS problem is
\[
\Omega{\left(\frac{\sqrt{\log{n}}}{(\log{\log{n}})^{3/2}}\right)}.
\]
This lower bound holds even when the outputs are succinctly encoded.
\end{thm}
A property of both these problems is that any two consecutive outputs
differ by at most one. This allows each output to be specified in
a constant number of bits. The restriction on the word and input
alphabet size in Theorem~\ref{thm:lcs} derives directly from the
ability to encode the output succinctly and its necessity will become
clearer when we describe the proof technique.

It is at first tempting to believe that there may be a simple reduction
between the edit distance, LCS and Hamming distance problems which
will allow us to derive lower bounds without requiring further work.
Although such direct reductions appear elusive in our streaming setting,
we are able to exploit more subtle and indirect relationships between
the distance measures to obtain our lower bounds.

We complement our edit distance lower bound with a new cell-probe
algorithm which runs in polylogarithmic time per arriving symbol. 
\begin{thm}
\label{thm:edit-upper} In the cell-probe model with any word size
$w$ and alphabet size that is polynomial in $n$, the edit distance
problem can be solved in 
\[
O{\left(\frac{\log^{2}n}{w}\right)}
\]
amortised time per output.
\end{thm}
The fastest RAM algorithm for online edit distance runs in $O(n)$ time by simply
adding a new column to the classic dynamic programming matrix for each new arrival. Despite the computation of edit distance being a widely studied topic,  it is not at all clear how one can do significantly better than this naive approach given the dependencies which seem to be inherent in the standard dynamic programming formulation of the problem.  We believe it is therefore of independent interest for those studying the edit distance that there is now an exponential gap between
the known RAM and cell-probe complexities. It is still unresolved
whether a similarly fast cell-probe algorithm exists for the online LCS problem.

\subsection{Prior work}

The field of streaming algorithms is well studied and the specific
question of how efficiently to find patterns in a stream is a fundamental
problem that has received increasing attention over the past few years.
In a classic result of Galil's~\cite{Galil:1981} from the early
1980s, exact matching was shown to be solvable in constant time per
arriving symbol in a stream. Nearly 30 years later a general online-to-offline
reduction was shown which enables many offline algorithms to be made
online with a worst case logarithmic factor overhead in the time complexity
per arriving symbol in the stream~\cite{CEPP:2008}. For some problems,
this reduction gives us the best time complexity known, but in other
cases it is possible to do even better. One example where a more efficient
online algorithm exists is the $k$-mismatch problem, which has an
$O(n\sqrt{k\log k})$ complexity offline but $O(\sqrt{k}\log{k}+\log{m})$
time per new arriving symbol online~\cite{CS:2010}, where $n$ is
the length of the text and $m$ is the length of the pattern.

The online-to-offline reduction of~\cite{CEPP:2008} does not however
lend itself to problems that can be described as non-local. The distance
function between two strings of the same length is said to be local
if it is the sum of the disjoint contribution from each aligned pair
of symbols. For example, Hamming distance is a local distance function
whereas edit distance is non-local. Streaming pattern matching upper
bounds have also been developed for a range of non-local problems,
including function matching, parameterised matching, swap distance,
$k$-differences as well as others~\cite{CS:2009}. These algorithms
have necessarily been particular to each distance function.

Given the rarity of constant time streaming algorithms, it is therefore
natural to ask for lower bounds; what are the limits of how fast a
streaming pattern matching problem can be solved? The first steps
towards an answer to this question were given in~\cite{CJ:2011},
where a lower bound for convolution, or equivalently the cross-correlation,
was given. Computing cross-correlations is an important component
of many of the fastest pattern matching algorithms. By giving a lower
bound for the time to perform the cross-correlation, a lower bound
is therefore provided for a whole class of pattern matching algorithms.
However the question of whether a particular pattern matching problem
could be solved by some other faster means remained open. In~\cite{CJS:2013},
the three authors of this paper gave the first lower bound for a pattern
matching problem. The streaming Hamming distance problem was shown
to have a logarithmic lower bound when the input alphabet size is
sufficiently large. This provided the first separation between two
pattern matching problems: exact matching, which can be solved in
constant time, and Hamming distance which cannot.


\subsection{The cell-probe model}

Our bounds hold in the \emph{cell-probe model} which is a particularly
strong model of computation, introduced originally by Minsky and Papert~\cite{MP:1969}
in a different context and then subsequently by Fredman~\cite{Fredman:1978}
and Yao~\cite{Yao1981:Tables}. The generality of the cell-probe
model makes it attractive for establishing lower bounds for dynamic
data structure problems, and many such results have been given in
the past couple of decades. The approaches taken had historically
been based only on communication complexity arguments and the chronogram
technique of Fredman and Saks~\cite{FS1989:chronogram}, which until
recently were able to prove $\Omega(\log{n}/\log{\log{n}})$ lower
bounds at best. There remains however, a number of unsatisfying gaps
between the lower bounds and known upper bounds. However, in 2004,
a breakthrough led by \Patrascu and Demaine gave us the tools to
seal the gaps for several data structure problems~\cite{PD2006:Low-Bounds}
as well as giving the first $\Omega(\log{n})$ lower bounds. This
new technique is based on information theoretic arguments that we
also employ here.


In the cell-probe model there is a separation between the computing
unit and the memory, which is external and consists of an (unbounded)
array of cells of $w$ bits each. The computing unit has no internal
memory cells of its own. Any computation performed is free and may
be non-uniform. The cost of processing an update or query (in our
case outputting the answer when a value in the stream arrives) is
the number of distinct cells accessed (cell-probes) during that update,
or query. This general view makes the model very strong. In particular,
any lower bounds in the cell-probe model hold in the word-RAM model
(with the same cell size). In the word-RAM model, certain operations
on words, such as addition, subtraction and possibly multiplication,
take constant time (see for example~\cite{Hagerup:1998} for a detailed
introduction). Although in our case we place no minimum size restriction
on the size of a word, much of the previous work has required that
words are sufficiently large to be able to store the address of any
cell of memory. When the word size $w=1$ then our new lower bounds
also hold, for example, for the weaker multi-head Turing machine model
with a constant number of heads.




\subsection{Technical contributions}

\label{sec:contributions}


\paragraph{New lower bounds} One of the most important techniques for online lower bounds is based
on the \emph{information transfer method} of \Patrascu and Demaine~\cite{PD2006:Low-Bounds}.
For a pair of time intervals, the information transfer is the set
of memory cells that are written during the first interval, read in
the next and not overwritten in between. These cells must contain \emph{all} the
information  from
the updates during the first interval that the algorithm needs in order 
to produce correct outputs in the next interval. If
one can prove that this quantity is large for many pairs of intervals
then the desired lower bounds follow. To do this we relate the size
of the information transfer to the conditional entropy of the outputs
in the relevant time interval. The main task of proving lower bounds
reduces to that of devising a hard input distribution for which outputs
have high entropy conditioned on selected previous values of the input.

Previous applications of the information transfer technique have required
that the word size $w$ is $\Omega(\log n)$~\cite{PD2006:Low-Bounds,CEPP:2011,CJS:2013}.
To circumvent this limitation we have developed a new encoding of
the information transfer that is efficient for arbitrarily small values
of $w$, in particular $w=1$ as in the bit-probe model. The overall
method is to combine an encoding based on cell addresses
with a new encoding that identifies a cell with the time step at which
it is read. This combination of two encodings enables us to prove
new lower bounds for the convolution and Hamming distance problems.
Moreover it is a crucial first step in developing our new lower bounds
for the edit distance and LCS problems where we restrict our attention
to constant sized alphabets.

The edit distance and LCS problems raise a number of challenges not
presented by either of the other two problems we consider. These distance
measures are what we call \emph{non-local}. Focusing on the LCS problem
for the moment we can see that whether position $i$ of the fixed
array is included in the LCS or not depends not only on the value
in the stream that is aligned with $i$ but also on other values in
the stream. The information transfer technique has previously not
been applied to such non-local streaming problems. The main technical
difficulty that non-local distance measures introduce is a blurring
of the borders between intervals.

We describe a hard input distribution of the LCS problem for which the information
transfer technique is indeed applicable. The idea is to construct
a fixed array and a random input stream such that at many alignments,
from the length of the LCS one can obtain the Hamming distance between
the fixed array and the corresponding portion of the stream. 
In order to then apply the information
transfer technique we must prove that this direct relationship between
the length of the LCS and Hamming distance occurs with sufficiently
large probability. This is one of the more technical parts of the
paper. Once we have obtained a lower bound for the LCS problem we show that
the same lower bound holds also for the edit distance problem. This
follows from our LCS hard distribution combined with a squeezing lemma
that forces the edit distance to equal the Hamming distance at certain
alignments. 

\paragraph{New upper bound} Our cell-probe algorithm for edit distance is a non-trivial modification
of the classic dynamic programming solution which allows us to take
advantage of the fact that computation is free in the cell-probe model. We exploit the relationship
between edit distance and shortest paths in a directed acyclic graph
(DAG).  The nodes in the graph form a lattice  and the current edit distance is the
shortest path from the top-left node to the bottom-right node. Each new symbol
that arrives simply adds a new column of distances. This update operation is however slow even
in the cell-probe model; just writing the new values to memory
requires $\Theta(n/w)$ cell probes.
Our algorithm circumvents this problem.

The first key difference between our new method and the naive
approach is that instead of maintaining only the values of the latest column,
we maintain values from the $n$ latest columns.
We maintain
values denoted $D(j,i)$, where $D(j,i)$ is the shortest path
from the top-left to the node $(j,i)$ in the DAG over paths that are forced
to go via selected previous nodes. Further, we do not in fact maintain $D(j,i)$
for all rows~$j$, potentially leaving gaps in the table. As a result we can efficiently maintain the $D(j,i)$ values.
Despite the fact that our algorithm does not correctly compute the whole dynamic programming table we are able to show that for all~$i$, the outputted value $D(n-1,i)$ still is the correct edit distance after symbol $S[i]$ has arrived.

\subsection{Organisation}

In Section~\ref{sec:preliminaries} we set up some basic notation
and give problem definitions. In Section~\ref{sec:proofs} we describe
how to obtain the lower bounds. This section contains some key lemmas
which are solved separately in subsequent sections. In Section~\ref{sec:lcs}
we describe the hard distribution for the edit distance and the LCS problems.
In Section~\ref{sec:main} we explain the new encoding scheme that
we use with the information transfer method. Finally, in Section~\ref{sec:edit-upper}
we give the cell-probe algorithm that solves the edit distance problem.
This section can be read in isolation and does not build on previous
sections.

\section{Basic setup for the lower bounds\label{sec:preliminaries}}

In this section we introduce notation and concepts that are used heavily
in the lower bound proofs. We also formally define the streaming problems
in our new notation.

\subsection{Basic notation}

For a positive integer $n$, $[n]$ denotes the set $\{0,\dots,n-1\}$.
For an array $A$ of length $n$ and $i,j\in[n]$, we write $A[i]$
to denote the value at position $i$, and where $j\geq i$, $A[i,j]$
denotes the $(j-i+1)$-length subarray of $A$ starting at position
$i$. All logarithms are in base two and we assume that $n\geq4$
throughout.

We define a \emph{streaming problem} as follows. There is a \emph{fixed~array}
$F$ of length $n$ and an array $S$ of length $3n$, which is referred
to as the \emph{stream}. Both $F$ and $S$ are over the set $[2^{\delta}]$
of integers, referred to as the \emph{alphabet}, where $\delta$ is
a positive integer and is a parameter of the problem. An element of
the alphabet is often referred to as a \emph{symbol}. We let $t\in[n]$
denote the arrival time, or simply \emph{arrival} of the symbol $S[2n+t]$.
That is, for $t=0$, just before the symbol $S[2n]$ arrives, the
stream already contains $2n$ symbols. To capture the concept of a
data stream, not all symbols of $S$ are immediately available. More
precisely, just after arrival $t\in[n]$ only the symbols of $S[0,2n+t]$
are known, and importantly, the symbols $S[(2n+t+1),(3n-1)]$ are
not known. That is one new symbol is revealed at a time. We define
\[
S_{t}=S[(n+1+t),(2n+t)]
\]
to denote latest $n$ symbols of the stream up to arrival $t$. Once
the symbol at arrival~$t$ is revealed, and before the next symbol
at arrival $t+1$ is revealed, a function of $F$ and $S[0,2n+t]$
is computed and its value outputted.
We let the $n$-length array $Y$ denote 
the outputs such that $Y[t]$ is outputted immediately after arrival~$t$.
The outputs depend on which
streaming problem is considered:
\begin{itemize}
\item In the \textbf{Hamming distance} problem, $Y[t]=\HAMt$, which
is number of positions $i\in[n]$ such that $F[i]\neq S_{t}[i]$. 
\item In the \textbf{convolution} problem, $Y[t]=\sum_{i\in[n]}F[i]\cdot S_{t}[i]$. 
\item In the \textbf{edit distance} problem, $Y[t]=\min_{i\in[2n+1+t]}\ED{F,S[i,2n+t]}$. 
\item In the \textbf{longest common subsequence (}\textbf{\emph{LCS}}\textbf{)}
problem, $Y[t]=\LCSt$, which is the length of the LCS of
$F$ and $S_{t}$. 
\end{itemize}

\subsection{Information transfer and more notation \label{sec:more-notation}}

Our lower bounds hold for any randomised algorithm on its worst case
input. The approach to obtain such bounds is by applying \emph{Yao's
minimax principle}~\cite{Yao1977:Minimax}. That is, we show that
the lower bounds of Theorems~\ref{thm:conham} and~\ref{thm:lcs}
hold for any deterministic~algorithm on some random~input. This
means that we will devise a fixed array $F$ and describe a probability
distribution for the stream~$S$. We then show a lower bound on the expected
running time over $n$ symbol arrivals in the stream that holds
for any deterministic algorithm. Due to the minimax principle, the
same lower bound must then hold for any randomised~algorithm on its
worst~case input. The amortised bound is obtain by dividing by $n$.
From this point onwards, we consider an arbitrary deterministic algorithm
running with some fixed array $F$ on a random stream $S$. As it is used to show
a lower bound, such an $F$ and distribution on $S$ is referred to as a 
\emph{hard distribution}.


The \emph{information transfer tree}, denoted $\calT$, is a balanced
binary tree over $n$ leaves. To avoid technicalities we assume that
$n$ is a power of two. For a node $v$ of $\calT$, we let $\ell_{v}$
denote the number of leaves in the subtree rooted at $v$. The leaves
of $\calT$, from left to right, represent the arrival $t$ from $0$
to $n-1$. An internal node $v$ is associated with three arrivals,
$t_{0}$, $t_{1}$ and $t_{2}$. Here $t_{0}$ is the arrival represented
by the leftmost node in subtree rooted at $v$, similarly $t_{2}=t_{0}+\ell_{v}-1$ is the rightmost such node and $t_{1}=t_{0}+\ell_{v}/2-1$ is in the middle. That is, the intervals $[t_{0},t_{1}]$
and $[t_{1}+1,t_{2}]$ span the left and right subtrees of $v$, respectively.
We define the subarray $\Xv=S[2n+t_{0},2n+t_{1}]$ to represent the $\ell_{v}/2$ stream symbols
arriving during the arrival interval $[t_{0},t_{1}]$, and we define the subarray $\Yv=Y[t_{1}+1,t_{2}]$ to represent the $\ell_{v}/2$
outputs during the arrival interval $[t_{1}+1,t_{2}]$. We define $\Xvknown$ to be the concatenation of $S[0,(2n+t_{0}-1)]$
and $S[(2n+t_{1}+1),(3n-1)]$. That is, $\Xvknown$ contains all symbols
of $S$ except for those in $\Xv$.

When $\Xvknown$ is fixed to some constant $\Xvfix$ and $\Xv$ is
random, we write $H(\Yv\mid\Xvknown=\Xvfix)$ to denote the conditional
entropy of $\Yv$ under the fixed $\Xvknown$.

We define the \emph{information transfer} of a node $v$ of $\calT$,
denoted $\vitset$, to be the set of memory cells $c$ such that $c$
is written during the interval $[t_{0},t_{1}]$, read at some time
in $[t_{1}+1,t_{2}]$ and not overwritten in $[t_{1}+1,t_{2}]$ before
being read for the first time in $[t_{1}+1,t_{2}]$. The cells in
the information transfer $\vitset$ therefore contain, by definition,
all the information about the values in $\Xv$ that the algorithm
uses in order to correctly produce the outputs $\Yv$. By adding up
the sizes of the information transfers $\vitset$ over all nodes $v$
of $\calT$, we get a lower bound on the total running time (that is, the number
of cell reads). To see this, it is important to make the observation
that a cell read of $c\in\vitset$ at arrival $t$ will be accounted
for only in the information transfer of node $v$ and not in the information
transfer of a node $v'\neq v$.  As a shorthand for the size of
the information transfer, we define $\vit=|\vitset|$. Our aim is
to show that $\vit$ is large in expectation for a substantial proportion
of the nodes $v$ of $\calT$.

\section{Overall proofs of the lower bounds \label{sec:proofs}}

In this section we give the overall proofs for the main lower bound
results of Theorems~\ref{thm:conham} and~\ref{thm:lcs}. Let $v$
be any node of $\calT$. Suppose that $\Xvknown$ is fixed but the
symbols in $\Xv$ are randomly drawn in accordance with the distribution
on $S$, conditioned on the fixed value of $\Xvknown$. This induces
a distribution on the outputs $\Yv$. If the entropy of $\Yv$ is
large, conditioned on the fixed $\Xvknown$, then any algorithm must
probe many cells in order to produce the outputs $\Yv$, as it is
only through the information transfer $\vitset$ that the algorithm
can know anything about $\Xv$. We will soon make this claim more
precise. We first define a high-entropy node in the information transfer tree.

\begin{defn}
[\textbf{High-entropy node}]\label{def:high-node}A node $v$ in
$\calT$ is a \emph{high-entropy~node} if there is a positive constant
$k$ such that for \emph{any} fixed $\Xvfix$, 
$$
H(\Yv\mid\Xvknown=\Xvfix)\,\geq\, k\cdot\delta\cdot\ell_{v}.
$$
\end{defn}

To put this bound in perspective, note that the maximum conditional
entropy of $\Yv$ is bounded by the entropy of $\Xv$, which is at
most $\delta\cdot(\ell_{v}/2)$ and obtained when the values of $\Xv$
are independent and uniformly drawn from $[2^{\delta}]$. Thus, the
conditional entropy associated with a high-entropy node is the highest
possible up to some constant factor. As we will see, this is good
news for proving lower bounds, and for the Hamming distance and convolution
problems we have many high-entropy nodes. The following fact tells us that there are inputs with many high-entropy nodes.

\begin{lem}
\label{lem:high-entropy} For both the Hamming distance and convolution
problems, where outputs are given modulo $2^{\delta}$, there exists
a hard distribution and a constant $c>0$ such that $v\in\calT$ is a
high-entropy node if $\ell_{v}\geq c\cdot\sqrt{n}$.
\end{lem}
\begin{proof}
For the convolution problem, the lemma is equivalent to Lemma~2 of
Clifford and Jalsenius~\cite{CJ:2011}, although notations differ
and here we only consider nodes $v$ such that $\ell_{v}\geq c\cdot\sqrt{n}$. 

For the Hamming distance problem, the statement of the lemma is equivalent
to Lemma~2.2 of Clifford, Jalsenius and Sach~\cite{CJS:2013} with
the only difference that in our lemma above we give outputs modulo
$2^{\delta}$. In the previous work of \cite{CJS:2013}, $2^{\delta}\in\Theta(n)$,
but here we consider any arbitrary $\delta$. For this reason it is
not obvious that Lemma~2.2 of~\cite{CJS:2013} applies under the
modulo constraint. However, by inspection of the details in~\cite{CJS:2013},
we see that every output is given within a range of size~$2^{\delta}$,
hence the lemma is indeed applicable also with the modulo constraint.
\end{proof}

For the edit distance and LCS problems on the other hand, the maximum conditional
entropy of $\Yv$ is at
most $O(\ell_{v})$, independent of $\delta$. This
is because the outputs can be encoded succinctly. Therefore we cannot 
expect to obtain high-entropy nodes for these problems in general. Moreover, it is not
even clear whether high-entropy nodes can be obtained for constant $\delta$. For these two problems
we therefore rely on what we call medium-entropy nodes.
\begin{defn}
[\textbf{Medium-entropy node}]\label{def:medium-node}A node $v$
in $\calT$ is a \emph{medium-entropy~node} if there is a positive
constant $k$ such that for \emph{at~least~half} of the values $\Xvfix$
of $\Xvknown$ that have non-zero probability in the distribution
for $S$, 
$$
H(\Yv\mid\Xvknown=\Xvfix)\,\geq\,\frac{k\cdot\ell_{v}}{\sqrt{\log n\cdot\log\log n}}.
$$ 
\end{defn}

We show that there exists a hard distribution for the edit distance and LCS problems for which most nodes in the information transfer tree are medium-entropy.  This is the main technical result needed to establish our edit
distance and LCS lower bound. The proof of the following lemma is discussed in Sections~\ref{sec:lcs} and~\ref{sec:lcs-ham-relationship}.
\begin{lem}
\label{lem:medium-entropy} For the edit distance and LCS problems
with $\delta=2$, there exists a hard distribution such that $v\in\calT$
is a medium-entropy node if $\ell_{v}\geq2\sqrt{n}$.
\end{lem}

\begin{defn}
[\textbf{Fast node}]
Let $\vtime$ denote the number of cell reads that take
place during the interval $[t_{1}+1,t_{2}]$ represented by the right
subtree of a node~$v$. That is, $\vtime$ is the total number of
cell reads performed by the algorithm while outputting the values
in~$\Yv$. We say that the node $v$ of $\calT$ is \emph{fast} if
$$\expected{\vtime}~\leq~\ell_{v}\cdot\delta\cdot\log n.$$
\end{defn}
For nodes that are both fast and high or medium entropy, the following result gives us a lower bound on the expected information transfer for our different problems.  This is also  one of the main technical
contributions of the paper.
The proof is outlined in Section~\ref{sec:main}.

\begin{lem}
\label{lem:encoding} For the hard distributions of both the Hamming distance
and convolution problems with $\delta\geq1$ and $w\geq1$, for any fast high-entropy node
$v$ of $\calT$, 
$$
\expected{\vit} \in \Omega{\left(\frac{\delta\cdot\ell_{v}}{w+\log\log n}\right)}.
$$
For the hard distributions of both the edit distance and LCS problems with
$\delta=2$ and $w=1$, for any fast medium-entropy node $v$ of $\calT$, 
$$
\expected{\vit} \in \Omega{\left(\frac{\ell_{v}}{\sqrt{\log n}\cdot(\log\log n)^{3/2}}\right)}.
$$

\end{lem}

\subsection{Obtaining the cell-probe lower bounds}

We now prove the main lower bound Theorems~\ref{thm:conham} and~\ref{thm:lcs} 
using
Lemmas~\ref{lem:high-entropy}, \ref{lem:medium-entropy} and~\ref{lem:encoding} from above.
Consider a hard distribution 
of the Hamming distance or convolution problems that satisfies 
Lemma~\ref{lem:high-entropy}.
Let $\calA$ be any deterministic algorithm that solves the problem.
Let $T$ denote the total running time of $\calA$ over $n$ arriving
values. The proof continues under the assumption that
$$\expected{T}\,\leq\,\frac{1}{2}\cdot\delta\cdot n\cdot \log n,$$
otherwise the lower bound of Theorem~\ref{thm:conham} is already established.

\global\long\def\dlow{\ensuremath{d_{\textup{low}}}}

Let $\dlow\in[\log n]$ be the smallest distance from the root of
the tree $\calT$ such that $\ell_{v}\leq c\cdot\sqrt{n}$, where
$c$ is the constant in the statement of Lemma~\ref{lem:high-entropy}
and $v$ is any node at depth $\dlow$. 
Let $V_{d}$ denote the set of nodes at distance $d$ from the root in~$\calT$.
Thus, by Lemma~\ref{lem:high-entropy}, for $d\in[\dlow]$, every node
$v\in V_{d}$ is a high-entropy node. Since the cell reads are disjoint over the nodes $v$ in $V_{d}$, we have
that $\sum_{v\in V_{d}}\vtime\leq T$. 
It follows from the linearity of expectation and the definition of a
fast node  that \emph{at least
half} of the nodes of $V_{d}$ are fast, otherwise $\expected{T}$
exceeds $\frac{1}{2}\delta n \log n$.

We can now sum the information transfer sizes $I_{v}$ over all fast
nodes in $V_{d}$ for every $d\in[\dlow]$. By applying Lemma~\ref{lem:encoding}
and linearity of expectation we get a lower bound of
$$
\frac{k'\cdot\delta\cdot n\cdot\log n}{w+\log\log n}
$$
on the expected total number of cell reads, where $k'$ is a constant
that depends on the constants from Lemmas~\ref{lem:high-entropy}
and~\ref{lem:encoding}, respectively. We divide by $n$ to get the
amortised lower bound of Theorem~\ref{thm:conham}. This concludes the lower bound proofs
for Hamming distance and convolution.

To prove the edit distance and LCS lower bounds of Theorem~\ref{thm:lcs},
we use the same argument but replace Lemma~\ref{lem:high-entropy}
with Lemma~\ref{lem:medium-entropy}, use the second half of 
Lemma~\ref{lem:encoding} and of course assume that $\delta=2$
and $w=1$.

\section{A hard distribution for the edit distance and LCS problems \label{sec:lcs}}

In this section and the next we prove Lemma~\ref{lem:medium-entropy} which says
that for the LCS and edit distance problems with $\delta=2$, there
exists a hard distribution such that a node $v\in\calT$ is a medium-entropy
node if $\ell_{v}\geq2\sqrt{n}$. 

\subsection{The hard distribution \label{sec:hard-instance}}

We begin by defining the hard distribution which is the same for
the edit distance and LCS problems. The alphabet has four symbols: $\pchar$, $\fchar$, $\bchar$ and
$\schar$. The symbol $\pchar$ is abundant in both $F$ and $S$
and always occurs in contiguous stretches of length $\pad$, where
we define
\[
\pad=4\sqrt{\log n\cdot\log\log n}.
\]
The symbols $\bchar$ and $\schar$ (short for \emph{heads} and \emph{tails})
represent coin flips, and $\fchar$ only occurs in $F$. 

We define $S$ to be of the form 
$$S=\pchar^{\pad}z_{1}\pchar^{\pad}z_{2}\pchar^{\pad}z_{3}\cdots,$$
where $\pchar^{\pad}$ denotes a stretch of $\pad$ $\pchar$-symbols, and
each $z_{i}$ is chosen independently and uniformly at random
from $\{\bchar,\schar\}$. That is one can obtain $S$ by flipping
$3n/(\pad+1)$ coins. For brevity, we assume that $\pad+1$ divides $n$.

We define $F$ to be of the form
$$F~=~\pchar^{\pad}\!\fchar\pchar^{\pad}\!\fchar\pchar^{\pad}\!\fchar\cdots,$$
 with the only exception that $\Theta(
\log n)$ of the $\fchar$-symbols
are replaced with the $\bchar$-symbol as follows. For every $j\in\{\sqrt{n},
\dots,n\}$
that is a power of two, identify the $\fchar$ in $F$ that is closest
to index $(n-j)$, breaking ties arbitrarily, and replace it with an $\bchar$-symbol. This concludes the
description of the hard distribution.

The purpose of repeated $\pchar^{\pad}$ substrings is to ensure a 
good probability that an LCS of $F$ and $S_t$ (the most recent $n$ symbols of $S$) 
omits no $\pchar$-symbols, 
enforcing a structure on the LCS. When this is the case we are able to use this 
structure to lower bound $H(\Yv\mid\Xvknown=\Xvfix)$. The value of $\pad$ has
been chosen carefully; a larger value of $\pad$ decreases the entropy of $S$, 
hence also the entropy of $\Yv$, and
a smaller value of $\pad$ increases the probability of the LCS omitting
$\pchar$-symbols.

\subsection{LCS and medium-entropy nodes \label{sec:LCS-medium}}

We now prove the LCS part of Lemma~\ref{lem:medium-entropy}. The edit distance part is proved in Section~\ref{sec:edit-medium}.
In the following we consider an arbitrary node $v\in\calT$ with 
$\ell_{v}\geq2\sqrt{n}$ and only consider the arrivals $t$ during which 
$Y_v$ is to be outputted. We now show that, using our hard distribution for the 
LCS problem with $\delta=2$, node $v$ has medium-entropy,
proving the LCS part of Lemma~\ref{lem:medium-entropy}.

We say that an arrival $t$ is \emph{\PA} if  $S_{t}[i]=\pchar$
whenever $F[i]=\pchar$. Hence $S_{t}[i]\in\{\bchar,\schar\}$
whenever $F[i]\in\{\fchar,\bchar\}$. These \PA arrivals are regular, 
occurring once in every $(\pad+1)$.  Let $\pav\subseteq[n]$ be the set of all \PA arrivals $t$ such
that $t$ is an arrival in the second~half of the arrival interval
during which $\Yv$ is outputted. More precisely, using notation from
Section~\ref{sec:more-notation} where the information transfer tree
$\calT$ was defined, $\pav$ is the set of \PA arrivals $t$ such
that $t\in[(t_{1}+1+\ell_{v}/4),t_{2}]$, where $v$ is a node in
the tree $\calT$. Hence $|\pav|=(\ell_{v}/4)/(\pad+1)$.

In the following lemma we will see that if
we know the Hamming distance, $\HAMt$ when $t$ is \PA then we can infer symbols from
the unknown inputs in $\Xv$. This fact follows from the observation that 
there is exactly one $\bchar$-symbol in $F$
which slides across $\Xv$ as $t$ increases.

\begin{lem}
\label{lem:recover-symbol} Consider a node $v$ of $\calT$ such
that $\ell_{v}\geq2\sqrt{n}$, and further that $\Xvknown=\Xvfix$ is known.
In the LCS hard distribution, for at least half of the \PA $t$, the value $\HAMt$ reveals
a non-$\pchar$ symbol in $\Xv$. No two distinct arrivals
reveal the same symbol in $\Xv$.\end{lem}

\begin{proof}
Suppose that $\ell_{v}\geq 2\sqrt{n}$ and consider any $t\in\pav$.
From the definition of the hard distribution it follows that there is
\emph{exactly one} index $i\in[n]$ such that  $i$
is an index of $S_{t}$ included in the substring $\Xv$ and $F[i]=\bchar$. Since $t$
is \PA, $S_{t}[i]\in\{\bchar,\schar\}$ and every other position
of $\Xv$ that holds a non-$\pchar$ symbol is aligned with a $\fchar$-symbol
of $F$. Since all elements of $S_{t}$ except for those in $\Xv$
are known, from the value of $\HAMt$ we can uniquely determine the
value of $S_{t}[i]$.

The second part of the lemma follows immediately as any two elements
of $\Xv$ that are determined at distinct \PA arrivals must
be at distinct positions of $\Xv$.\end{proof}

We can therefore directly infer that for the Hamming distance problem, under the LCS hard distribution,  $H(\Yv\mid\Xvknown=\Xvfix) \in \Omega(\ell_v/\rho)$ for all $\Xvfix$. This is because each of the $\Omega(\ell_{v}/\rho)$ non-$\pchar$ symbols in $\Xv$ corresponds to an independent coin-flip. 
In order to get the LCS lower bound we show in the following lemma that we can in fact often infer $\HAMt$ from $\LCSt$. The proof forms the technical core of the lower bound and Section~\ref{sec:lcs-ham-relationship} is devoted to~it.

\begin{lem}
\label{lem:LCSruns} In the LCS hard distribution, at any \PA arrival
$t$, with probability at least $9/10$, $\LCSt=n-\HAMt$. 
\end{lem}

Lemma~\ref{lem:LCSruns} is not sufficient, even in conjunction with Lemma~\ref{lem:recover-symbol}, to show that there are many medium-entropy nodes for the LCS problem. However we can use it to establish the following fact which says that for most $\Xvknown$ there is a \emph{fixed} subset of the \PA arrivals 
of size $\Omega(\ell_{v}/\rho)$ such that $\LCSt=n-\HAMt$ with probability at least $4/7$. We will see that this will then be sufficient to prove Lemma~\ref{lem:medium-entropy}.

\begin{lem}
\label{lem:lcs-ham} Suppose that $v$ is a node of $\calT$ such
that $\ell_{v}\geq2\sqrt{n}$. In the LCS hard distribution, for at least
half of the values of $\Xvknown$ there is a fixed set $\pav^{*}\subseteq\pav$
of \PA arrivals, where $|\pav^{*}|\geq|\pav|/15$, such that for
any $t\in\pav^{*}$, the probability that $\LCSt=n-\HAMt$ is at least
$4/7$. 
\end{lem}
\begin{proof}
Let $v$ be a node in the tree $\calT$ such that $\ell_{v}\geq2\sqrt{n}$.
First we claim that for at least half of the values of $\Xvknown$,
$\LCSt=n-\HAMt$ for a fraction of at least $4/5$ of all \PA arrivals
during the interval where $\Yv$ is outputted. Notice that these \PA
arrivals are not fixed. In particular, many of them might not be in
$\pav$. We may assume that they depend arbitrarily on the values
of $S$, that is both $\Xvknown$ and $\Xv$. We show the claim by
contradiction. Under the assumption that the claim is false we will
maximise the total number of arrivals at which the LCS output equals
$n$ minus the Hamming distance and see that this will contradict
Lemma~\ref{lem:LCSruns}.

So, suppose that the claim is false. This means that fewer than half
of the $\Xvknown$ values have the property that $\LCSt=n-\HAMt$
for any arbitrary number of \PA arrivals $t$, and for the remaining
values of $\Xv$, $\LCSt=n-\HAMt$ for less than a fraction of $4/5$
of the \PA arrivals $t$. Thus, over all values of $S$, the fraction
of \PA arrivals~$t$ during the interval where $\Yv$ is outputted,
for which $\LCSt=n-\HAMt$, is less than $(1/2)\cdot1+(1/2)\cdot(4/5)=9/10$.
Observe that for any fixed \PA arrival $t$, any two values of the
substring $S_{t}$ are counted the same number of times over all values
of $S$. That is, a particular substring $S_{t}$ does not occur more
frequently than any other substring. Thus, assuming the claim is not
true contradicts Lemma~\ref{lem:LCSruns}.

We now know that for at least half of the values of $\Xvknown$, a
fraction of at least $4/5$ of all \PA arrivals $t$ have the property
that $\LCSt=n-\HAMt$. Let $\pav'\subseteq\pav$ be the set of all
arrivals $t\in\pav$ such that $\LCSt=n-\HAMt$. Recall that during the interval where $\Yv$ is outputted there is
a total of $2|\pav|$ \PA arrivals. Thus, 
\[
|\pav'|\,\geq\,\frac{4}{5}\cdot2|\pav|-|\pav|\,=\,\frac{3}{5}|\pav|.
\]
We will now argue that there must be a fixed choice $\pav^{*}\subseteq\pav$
of arrivals such that for every $t\in\pav^{*}$, the probability that
$\LCSt=n-\HAMt$ is at least $4/7$, and $|\pav^{*}|\geq(1/15)|\pav|$.
This would conclude the proof of the lemma. Note that $\pav^{*}$
does not depend on $S$, as opposed to $\pav'$ which may depend on
$S$. Again we will use proof by contradiction.

Let $\Xvknown$ be a value such that at $|\pav'|\geq(3/5)|\pav|$
and suppose that there is no $\pav^{*}$ with the above property.
To show contradiction we will show that $|\pav'|<(3/5)|\pav|$. Under
the assumption that there is no $\pav^{*}$ with the above property,
we may suppose that just under $(1/15)|\pav|$ fixed arrivals $t\in\pav$
have the property that $\LCSt=n-\HAMt$ with probability~$1$, and
for the remaining arrivals $t\in\pav$, the probability that $\LCSt=n-\HAMt$
is just below $4/7$. In the hard distribution, the induced distribution
for $\Xv$ conditioned on any fixed $\Xvknown$ is uniform, hence
\[
|\pav'|\,<\,\frac{1}{15}|\pav|\cdot1+(1-\frac{1}{15})|\pav|\cdot\frac{4}{7}\,=\,\frac{3}{5}|\pav|,
\]
which is the contradiction we were looking for.
\end{proof}

We can now complete the proof of the LCS part of Lemma~\ref{lem:medium-entropy}.
\begin{proof}
[Putting the pieces together to complete the proof of Lemma~\ref{lem:medium-entropy}]Let $v$
be a node in the tree $\calT$ such that $\ell_{v}\geq2\sqrt{n}$.
Let $\Xvknown$ be  such that there is a set $\pav^{*}\subseteq\pav$
of arrivals satisfying Lemma~\ref{lem:lcs-ham}. For each
arrival $t$ in $\pav^{*}$ we can infer $\HAMt$ from the output
with probability at least $4/7$, hence by applying Lemma~\ref{lem:recover-symbol}
we can determine an element of $\Xv$ with probability at least $4/7$.
We may not know when an element is correctly identified, but nevertheless,
in $4$ out of $7$ cases the output at arrival $t$ will reflect
some element of $\Xv$, which was picked according to a coin flip.
To assume the worst, we may assume that whenever the element is an
$\bchar$-symbol, the output reflects the value of the element. Further,
whenever the output does not reflect the value of the element, we
may assume that it is an $\bchar$-symbol as well. Thus, the contribution
to the conditional entropy of $\Yv$ from the output at arrival $t\in\pav^{*}$
is at least the entropy of a biased coin for which one side has probability
$4/7-1/2=1/14$. The (binary) entropy of such a coin is bounded from below
by $0.37$.

Thus, the total contribution to the conditional entropy of $\Yv$
from all arrivals in $\pav^{*}$ is at least 
\[
0.37\cdot|\pav^{*}|~\geq~0.37\cdot\frac{1}{15}|\pav|~=~\frac{0.37}{15}\cdot\frac{\ell_{v}}{4(\pad+1)}~=~\frac{0.37}{15}\cdot\frac{\ell_{v}}{4(4\sqrt{\log n\cdot\log\log n}+1)}.
\]
This value matches the definition of a medium-entropy node for a suitable
constant~$k$, which concludes the proof of the LCS part of Lemma~\ref{lem:medium-entropy}.
\end{proof}

\subsection{Edit distance and medium-entropy nodes \label{sec:edit-medium}}

To show the edit distance part of Lemma~\ref{lem:medium-entropy} we use the 
same argument and hard distribution as for the LCS problem, coupled with the following squeezing property. 

\begin{lem}
\label{lem:squeeze} For any $F$ and $S_t$, 
$$n-\LCSt \,\leq\, \min\limits_{j\in[2n+1+t]}\ED{F,S[j,2n+t]} \,\leq\, \HAMt.$$
\end{lem}
\begin{proof}
The second inequality follows immediately since the Hamming distance is a restricted version of edit distance.
We now focus on the first inequality.

Since the lemma makes no assumption on the strings $F$ and $S$,
we will prove the following equivalent statement where we align $F$
with prefixes of a longer string instead of suffixes. Let $G$
be any string of length $2n$. We will show that 
\[
n-\LCS(F,G[0,n-1])~\leq~\min_{j\in[2n]}\ED{F,G[0,j]}
\]

Let $j^{*}$ be a $j$ that minimises the right hand side of the inequality.
Thus, we want to show that $n-\LCS(F,G[0,n-1])\leq\ED{F,G[0,j^{*}]}$.
We consider two cases, depending on the value of $j^{*}$.

In the first case, suppose that $j^{*}<n-1$. Here we think of the
edit distance as the number of edit operations (replace, insert, delete)
required to transform $F$ into $G[0,j^{*}]$. We say that an index~$i\in[n]$
of $F$ is \emph{untouched} if $F[i]$ is not subject to an edit operation,
that is, neither replaced nor deleted. Let $u$ be the number of untouched
indices. The number of edit operations required to transform $F$
into $G[0,j^{*}]$ is at least $n-u$, where we have equality if there
are no insertions. The symbols at the untouched indices of $F$ make
a common subsequence of $F$ and $G[0,j^{*}]$, hence $n-\LCS(F,G[0,n-1])$
is at most $n-u$. This concludes the first case.

In the second case, suppose that $j^{*}\geq n-1$. Similarly to above,
let $u$ be the number of untouched indices in $[j^{*}+1]$ when transforming
$G[0,j^{*}]$ into $F$. Let $u'$ be the number of untouched indices
$i$ such that $i\leq n-1$, hence $\LCS(F,G[0,n-1])\geq n-u'$. Thus,
\begin{align*}
n-\LCS(F,G[0,n-1])\leq n-u' & =(j^{*}+1)-(u'+j^{*}-(n-1))\\
 & \leq(j^{*}+1)-u\\
 & \leq\ED{F,G[0,j^{*}]}.\tag*{\qedhere}
\end{align*}

\end{proof}


We can now give the proof of the edit distance part of Lemma~\ref{lem:medium-entropy}.

\begin{proof}
[Proof of the edit distance part of Lemma~\ref{lem:medium-entropy}]By
combining Lemmas~\ref{lem:LCSruns} and~\ref{lem:squeeze}, we have
that with probability at least $9/10$, the edit distance equals $\HAMt$
at \PA arrivals $t$. Thus, Lemma~\ref{lem:lcs-ham} holds also
for the edit distance problem. Finally we use the same argument as
for the LCS part of the proof of Lemma~\ref{lem:medium-entropy}
from the previous section to conclude the proof of the edit distance
part of the lemma.
\end{proof}

\section{The relationship between LCS and Hamming distance \label{sec:lcs-ham-relationship}}

In this section we prove Lemma~\ref{lem:LCSruns} which says that
in the hard instance for the LCS problem, at any \PA arrival $t$,
$\LCSt=n-\HAMt$ with probability at least $9/10$.

The overall approach is to show that for a \PA arrival~$t$ there
is a high probability that any LCS of $F$ and $S_t$ (the latest $n$ symbols
of the stream) includes all $\pchar$-symbols from~$F$. When the LCS indeed includes all $\pchar$-symbols, 
showing that $\LCSt=n-\HAMt$ follows straightforwardly.

Let $t \in [n]$ be a \PA arrival.  Figure~\ref{fig:matrix}(a) illustrates an example of the alignment of $F$ and $S$.   For each index $i$ so that $F[i] = \bchar$, the rectangular box illustrates the subarray $S_t[i-(\pad + \pad^2),  i+(\pad + \pad^2)]$.  As the distance between two consecutive $\bchar$-symbols is at least $\sqrt{n}$ these boxes cannot overlap.  Now let us define $\calC_t$ to be the set of common subsequences of $F$ and $S_t$ so that each subsequence includes at most one $\bchar$-symbol from each rectangular box and no other $\bchar$-symbols. Observe that any subsequence of $F$ and $S_t$ is over the alphabet $\{\pchar, \bchar\}$. The set $\calC_t$ has the following very useful property.

%

\begin{lem}
\label{lem:LCS} In the hard distribution, for any \PA arrival $t$, the set of subsequences $\calC_t$ contains all longest common subsequences of $F$ and $S_t$.
\end{lem}
\begin{proof}
  We now argue that a common subsequence that includes an $\bchar$-symbol from outside the rectangular boxes would have to omit more $\pchar$-symbols than the total number of $\bchar$-symbols in $F$, hence it cannot be an LCS. Since all $\bchar$-symbol of $F$ are aligned with the middle element of a rectangular box, including an $\bchar$-symbol from outside a box means that at least $\pad\cdot\pad$ $\pchar$-symbols must be omitted. Now, $\pad^2>\log n$ and there are no more than $\log n$ $\bchar$-symbols in $F$.
  First observe there is always a common subsequence in $\calC_t$ which includes all the $\pchar$-symbols.
  
  Similarly, if a common subsequence includes more than one $\bchar$-symbol from the same rectangular box, one of them can be seen as being picked from outside another box, hence the same argument as above applies.
\end{proof}


\begin{figure}[t]
\centering{}\includegraphics{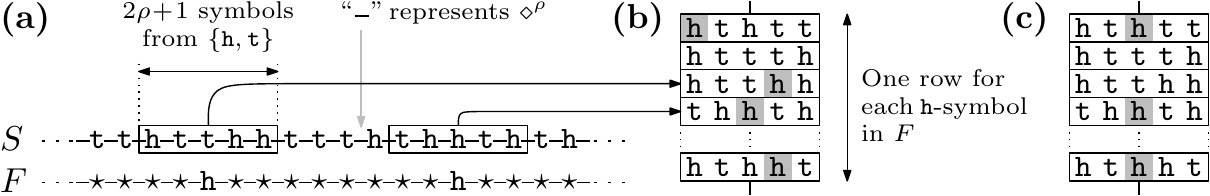} \protect\caption{An example of the 
alignment of $F$ and $S_t$ at some \PA arrival
$t$. Only a portion of the strings are shown, containing two occurrences
of the $\bchar$-symbol in $F$. Every substring $\pchar^{\pad}$ is illustrated
with a short line segment.\label{fig:matrix}}
\end{figure}


Consider now the rectangular boxes stacked on top of each other
as in Figure~\ref{fig:matrix}(b). The $i$-th box from the top corresponds
to the $i$-th $\bchar$-symbol in $F$. Each box contains $2\pad+1$
elements from $\{\bchar,\schar\}$ and may be regarded as a row of
a matrix with $2\pad+1$ columns. We refer to this matrix as $M_{t}$.
Observe that in the hard distribution, the entries of $M_{t}$ are drawn independently and uniformly at random
from $\{\bchar,\schar\}$.

\global\long\def\row{\ensuremath{\textup{row}}}


Define $\pi$ to be a set of matrix coordinates of $\bchar$-entries of $M_{t}$, containing
at most one entry from each row. Let $\Pi_t$ be the set of all such sets $\pi$. In Figure~\ref{fig:matrix}(b) and~(c)
we have illustrated two examples of $\pi$, where its elements 
are highlighted in grey.  Given some $\pi$, let $h_{1},\dots,h_{|\pi|}$ denote the
elements of $\pi$ in top-to-bottom order. Each set $\pi$ uniquely specifies a subsequence $C_{\pi} \in \calC_t$ by describing which $\bchar$-symbols are chosen from each rectangular box. The subsequence $C_{\pi}$ is then completed greedily by including as many $\pchar$-symbols as possible.

We also define $\operatorname{col}(i)\in \{1,\dots, 2\pad+1\}$ to be the column in which $h_i$ occurs and the value $n_{
  \pchar}$ to be the number of $\pchar$-symbols in $F$. Observe that $n_\pchar$ only depends on $n$. We can now define
\[
\len(\pi) = n_{\pchar} 
+ |\pi|
- \frac \pad 2 \cdot \Big(
|\pad+1- \operatorname{col}(1)|
+ |\pad+1- \operatorname{col}(|\pi|)|
+ \sum_{i=1}^{|\pi-1|} |\operatorname{col}(i) - \operatorname{col}(i+1)|
\Big).
\]
The following lemma tells us that  $\len(\pi)$ is in fact the length of the common subsequence $C_\pi$ implied by $\pi$.  



\begin{lem}
\label{lem:len-LCS} For any \PA arrival $t$ and any $\pi\in\Pi_t$, $\len(\pi)=|C_\pi|$.
 \end{lem}

\begin{proof}
Let $t$ be any \PA arrival and let $\pi$ be any set from $\Pi_t$.
As described above, $\pi$ uniquely specifies a common subsequence
$C_{\pi}$ of $F$ and $S_{t}$.
Recall that the elements of $\pi$ are denoted $h_{1},\dots,h_{|\pi|}$
in top-to-bottom order with respect to the matrix $M_{t}$.
For $i\in\{0,\dots,|\pi|\}$, we define
\[
\dist(i,i+1)=
\begin{cases}
  |\pad+1- \operatorname{col}(1)| & \textup{if $i=0$,}\\
  |\pad+1- \operatorname{col}(|\pi|)| & \textup{if $i=|\pi|$,}\\
|\operatorname{col}(i) - \operatorname{col}(i+1)| & \textup{otherwise.}
\end{cases}
\]
We can now write $\len(\pi)$ as
\[
\len(\pi) = n_{\pchar} 
+ |\pi|
- \frac \pad 2 \sum_{i=0}^{|\pi|} \dist(i,i+1).
\]
Referring to the above formula for $\len(\pi)$,
and starting with the number $n_\pchar$ of $\pchar$-symbols, we will show that adding the number of $\bchar$-symbols in $C_\pi$, that is adding the term $|\pi|$, and subtracting the
other terms of the equation will indeed equal $|C_{\pi}|$.
That is, we will show that the number of omitted $\pchar$-symbols in $C_\pi$ is exactly 
\[
\frac{\pad}{2}\cdot\sum_{i=0}^{|\pi|}\dist(i,i+1).
\]

Similarly to the definition of $\operatorname{col}(i)$, we define for $i\in\{1,\dots,|\pi|\}$,
$\row(i)$ to be the row of $M_t$ in which $h_i$ occurs.

Let $s_{1}$
be the number of $\pchar$-symbols in $S_{t}$ to the left of the $\bchar$-symbol in $S_t$ that correspond to $h_{1}$.
Similarly, let $f_{1}$ be the number of $\pchar$-symbols
in $F$ to the left of the $\row(h_{1})$-th $\bchar$-symbol.
Thus, 
\[
|f_{1}-s_{1}|=\pad\cdot\dist(0,1).
\]

Now, let 
\[
d_{1}=\begin{cases}
\phantom{-}1 & \mbox{if \ensuremath{f_{1}>s_{1}}},\\
-1 & \mbox{otherwise.}
\end{cases}
\]
Thus, $d_{1}=1$ whenever $|s_{1}-f_{1}|$ $\pchar$-symbols of $F$,
up to the $\row(h_{1})$-th $\bchar$-symbol, are omitted from $C_{\pi}$.
Otherwise, all $\pchar$-symbols of $F$ up to the $\row(h_{1})$-th
$\bchar$-symbol are included in $C_{\pi}$, and $d_{1}=-1$. The
purpose of $d_{1}$ will be clear shortly.

Now consider $h_{i}\in\pi$ for $i\in\{2,\dots,|\pi|\}$. We will
define $s_{i}$, $f_{i}$ and $d_{i}$ similarly to $s_{1}$, $f_{1}$
and $d_{1}$.
That is, let $s_{i}$ be the number of $\pchar$-symbols in $S_{t}$
between the $\bchar$-symbols that correspond to $h_{i-1}$ and $h_{i}$, respectively.
Let $f_{i}$ be the number
of $\pchar$-symbols in $F$ between the $\row(i-1)$-th and $\row(i)$-th
$\bchar$-symbols in $F$. Thus, 
\[
|f_{i}-s_{i}|=\pad\cdot\dist(i-1,i).
\]
Similarly to $d_{1}$, let 
\[
d_{i}=\begin{cases}
\phantom{-}1 & \mbox{if \ensuremath{f_{i}>s_{i}}},\\
-1 & \mbox{otherwise.}
\end{cases}
\]
We have $d_{i}=1$ if $|f_{i}-s_{i}|$ $\pchar$-symbols between the
$\row(i-1)$-th and $\row(i)$-th $\bchar$-symbols of $F$ are omitted
from $C_{\pi}$, and $d_{i}=-1$ if no such $\pchar$-symbols are
omitted.

Finally, in order to capture $\pchar$-symbols to the right of the
last $\bchar$-symbol in $C_\pi$, let $s_{|\pi|+1}$ be the number of
$\pchar$-symbols in $S_{t}$ to the right of the $\bchar$-symbol that corresponds to $h_{\pi}$, and let
$f_{|\pi|+1}$ be the number of $\pchar$-symbols in $F$ to the right
of the $\row(h_{\pi})$-th $\bchar$-symbol in $F$. We have 
\[
|f_{|\pi|+1}-s_{|\pi|+1}|=\pad\cdot\dist(|\pi|,|\pi|+1),
\]
and we let 
\[
d_{|\pi|+1}=\begin{cases}
\phantom{-}1 & \mbox{if \ensuremath{f_{|\pi|+1}>s_{|\pi|+1}}},\\
-1 & \mbox{otherwise.}
\end{cases}
\]

The number $n_\pchar$ of $\pchar$-symbols is the same in both $F$ and $S_{t}$
and is exactly 
\[
n_\pchar=\sum_{i=1}^{|\pi|+1}f_{i}=\sum_{i=1}^{|\pi|+1}s_{i}.
\]
Hence, 
\[
\sum_{i=1}^{|\pi|+1}(f_{i}-s_{i})=\sum_{i=1}^{|\pi|+1}d_{i}\cdot|f_{i}-s_{i}|=0,
\]
where we have used the definition of $d_{i}$ from above. Separating
into positive and negative terms, we have that 
\[
\sum_{i\,|\, d_{i}=1}|f_{i}-s_{i}|=\sum_{i\,|\, d_{i}=-1}|f_{i}-s_{i}|,
\]
which we use in the next equation. The total number of $\pchar$-symbols
that are omitted in $C_{\pi}$ is 
\begin{align*}
\sum_{i\,|\, d_{i}=1}|f_{i}-s_{i}| & =\frac{1}{2}\left(\sum_{i\,|\, d_{i}=1}|f_{i}-s_{i}|+\sum_{i\,|\, d_{i}=-1}|f_{i}-s_{i}|\right)\\
 & =\frac{1}{2}\sum_{i=1}^{|\pi|+1}|f_{i}-s_{i}|\\
 & =\frac{1}{2}\sum_{i=1}^{|\pi|+1}\pad\cdot\dist(i-1,i)\\
 & =\frac{\pad}{2}\sum_{i=0}^{|\pi|}\dist(i,i+1),
\end{align*}
which is what we wanted to show.
\end{proof}


Let $\pi^{*}\in\Pi_t$ be the set of coordinates of all $\bchar$-symbols
that appear in the middle column of $M_t$. See Figure~\ref{fig:matrix}(c) for an example.  The following probabilistic fact tells us that we can simply choose these symbols and still maximise $\len(\pi)$ with constant probability.
The proof follows by first showing that if every $(d \times 1)$-submatrix of $M_t$ contains between $d/2 - \pad/4$ and $d/2+\pad/4$ $\bchar$-symbols then for all $\pi \in \Pi_t, \len(\pi) \leq \len(\pi^*)$. At a high level,  when the number of $\bchar$-symbols for all $(d \times 1)$ submatrices is within this bound one can never compensate from the cost of deviating from the middle column. We then show that $M_t$ has this property  with probability at least $9/10$. 
 

\begin{lem}
\label{lem:random-matrix} Let $M_t$ be a random matrix whose
elements are chosen independently and uniformly at random from $\{\bchar,\schar\}$.
 With probability at least $9/10$,
$\len(\pi)\leq\len(\pi^{*})$ for all $\pi\in\Pi_t$.\end{lem}
\begin{proof}
We will drop the subscript $t$ from $M_t$ and $\Pi_t$ in the rest of the proof.
Let $M$ be a random binary matrix whose elements are chosen independent
and uniformly at random from $\{\bchar,\schar\}$. For any $\pi\in\Pi$,
let $h_{1},\dots,h_{|\pi|}$ be the elements of $\pi$ in top-to-bottom
order with respect to the matrix $M$. We may write $\len(\pi)$ as
\[
\len(\pi)~=~n_\pchar+\val{\pi},
\]
where 
\begin{equation}
\val{\pi}~=~|\pi|-\frac{\pad}{2}\cdot\sum_{i=0}^{|\pi|}\dist(i,i+1)\label{eq:val}
\end{equation}
and $\dist(i,i+1)$ was defined in the proof of Lemma~\ref{lem:len-LCS}.
Since $n_{\pchar}$ only depends on $n$,
we will show that with probability $9/10$,
$\val{\pi}\leq\val{\pi^{*}}$ for every $\pi\in\Pi$.

Let $\pi$ be any set in $\Pi$. There is a unique partition of $\pi$
into disjoint subsets, which we denote $\pi_{1},\dots,\pi_{m}$, such
that two elements $h_{j},h_{j'}\in\pi$, where $j<j'$, belong to
the same $\pi_{i}$ if and only if
$\operatorname{col}(j)=\operatorname{col}(j+1)=\cdots=\operatorname{col}(j')$, and further, for any
two distinct elements $h_{j}\in\pi_{i}$ and $h_{j'}\in\pi_{i'}$,
where $i<i'$, we have $j<j'$. As en example, the partition of $\pi^{*}$
contains only the set $\pi^{*}$ itself as all elements are from the
same column.

For $i\in\{1,\dots,|\pi|\}$, let $\row(i)$ be the row of $M$ in which $h_i$ occurs.
For any $\pi_{i}$ in the partition of $\pi$, let 
\begin{align*}
\toprow(\pi_{i}) & ~=~\min_{h_j\in\pi_{i}}~\row(j),\\
\bottomrow(\pi_{i}) & ~=~\max_{h_j\in\pi_{i}}~\row(j),\\
\column(\pi_{i}) & ~=~~\textup{the column of \ensuremath{M} from which the elements of \ensuremath{\pi_{i}} are.}
\end{align*}

We say that $M$ is \emph{balanced} if every $(d\!\times\!1)$-submatrix
of $M$ (one column wide and height~$d$) contains at least $d/2-\pad/4$
and at most $d/2+\pad/4$ $\bchar$-symbols. We will first show that
if $M$ is balanced then $\val{\pi}\leq\val{\pi^{*}}$ for every~$\pi\in\Pi$.
Then we will show that a random $M$ is balanced with probability
$9/10$.

So, suppose that $M$ is balanced. Let $n_\bchar$ be the number of $\bchar$-symbols in $F$, that is the height of $M$. By considering the entire middle column of $M$, we have
\begin{equation}
\frac{n_{\bchar}}{2}-\frac{\pad}{4}~\leq~|\pi^{*}|=\val{\pi^{*}}.\label{eq:valstar}
\end{equation}
Let $\pi$ be any set in $\Pi$
and let $\pi_{1},\dots,\pi_{m}$ be the subsets in the partition
of $\pi$.
For $i\in\{1,\dots,m\}$, suppose $\pi_{i}=\{h_{j_{i}},\dots,h_{j'_{i}}\}$.
We define the cost of $\pi_i$ to be
\[
\cost(\pi_{i})~=~\frac{\pad}{2}\cdot\sum_{k=j_{i}}^{j'_{i}}\dist(k,k+1).
\]
The value $\val{\pi}$ in Equation~(\ref{eq:val}) can
be split into $m+1$ terms of which $m$ terms correspond to the contribution
from the $m$ subsets $\pi_{i}$. That is, 
\[
\val{\pi}~=\;-\frac{\pad}{2}\cdot\dist(0,1)+\sum_{i=1}^{m}\big(|\pi_{i}|-\cost(\pi_{i})\big),
\]
Let the height of the subcolumn spanned by elements from $\pi_i$ be 
\[
d_{i}~=~\bottomrow(\pi_{i})-\toprow(\pi_{i})+1.
\]
Then, since $M$ is balanced, 
\[
|\pi_{i}|~\leq~\frac{d_{i}}{2}+\frac{\pad}{4}.
\]
Thus, 
\begin{align}
\val{\pi} & ~\leq\;-\frac{\pad}{2}\cdot\dist(0,1)+\sum_{i=1}^{m}\left(\frac{d_{i}}{2}+\frac{\pad}{4}-\cost(\pi_{i})\right)\notag\\
 & ~\leq~\frac{n_{\bchar}}{2}-\frac{\pad}{2}\cdot\dist(0,1)+\sum_{i=1}^{m}\left(\frac{\pad}{4}-\cost(\pi_{i})\right)\label{eq:valbound}
\end{align}
since $\sum_{i=1}^{m}d_{i}$ is at most the number of rows $n_{\bchar}$
of $M$.

In order to show that $\val{\pi}\leq\val{\pi^{*}}$, first consider
the case where $m=1$. Then there is only one set in the partition
of $\pi$ and it follows immediately by Equation~(\ref{eq:val})
that $\val{\pi}\leq\val{\pi^{*}}$.

In the case $m=2$, suppose first that $\column(\pi_{1})$ is the
middle column of $M$. Then an upper bound on Inequality~\ref{eq:valbound}
is given by 
\[
\frac{n_{\bchar}}{2}-\frac{\pad}{2}\cdot0+\left(\frac{\pad}{4}-\frac{\pad}{2}\right)~=~\frac{n_{\bchar}}{2}-\frac{\pad}{4}~\leq~\val{\pi^{*}},
\]
where the last inequality uses Inequality~(\ref{eq:valstar}). The
case where $\column(\pi_{1})$ is not the middle column of $M$ is
similar as $\dist(0,1)$ is then at least~1.

Finally, in the case where $m\geq3$, first observe that $\cost(\pi_{i}$)
is at least $\pad/2$ for all $\pi_{i}$ except for possibly $\pi_{m}$
if $\column(\pi_{m})$ is the middle column of $M$. An upper bound
on Inequality~\ref{eq:valbound} is given by 
\[
\frac{n_{\bchar}}{2}-\frac{\pad}{2}\cdot0+(m-1)\cdot\left(\frac{\pad}{4}-\frac{\pad}{2}\right)+\left(\frac{\pad}{4}-0\right)~=~\frac{n_{\bchar}}{2}-(m-2)\cdot\frac{\pad}{4}~\leq~\frac{n_{\bchar}}{2}-\frac{\pad}{4}~\leq~\val{\pi^{*}},
\]

We will now show that if $M$ is random, then with probability at
least $9/10$, $M$ is balanced. Consider a arbitrary $(d\!\times\!1)$-submatrix
of $M_{t}$. The number of $\bchar$-entries in this submatrix is
binomially distributed with parameter $(d,\frac{1}{2})$. Thus, by
Hoeffding's inequality we have that the probability that there are
less than $(d/2-\pad/4)$ $\bchar$-entries in this submatrix is at
most 
\[
e^{{\frac{-2(\pad/4)^{2}}{d}}}~<~e^{{\frac{-4\log n\cdot\log\log n}{\log n}}}~\leq~\frac{1}{\log^{4}n},
\]
where the second inequality follows because the height of $M$ is
at most $\log n$ and the value of $\pad=4\sqrt{\log n\cdot\log\log n}$.
By symmetry, this is also an upper bound on the probability that there
are more than $(d/2+\pad/4)$ $\bchar$-entries in the submatrix.
Thus, $2/\log^{4}n$ is an upper bound on the probability that the
number of $\bchar$-entries is not in the range $d/2\pm\pad/4$.

By the union bound over all $(d\!\times\!1)$-submatrices of $M$
it follows that the probability that $M$ is balanced is at least
$9/10$ for sufficiently large $n$. To see this, recall that the
width of $M$ is $2\pad+1\leq\log n$, hence there are at most $\log^{3}n$
submatrices of width~1.
\end{proof}

%


Finally we can prove that in the hard instance for the LCS problem, at any \PA arrival $t$,
$\LCSt=n-\HAMt$ with probability at least $9/10$ and hence establish Lemma~\ref{lem:LCSruns}.

\begin{proof}[Proof of Lemma~\ref{lem:LCSruns}]
 By combining Lemma~\ref{lem:LCS} with Lemma~\ref{lem:len-LCS} we have that $\max_{\pi \in \Pi_t} \len(\pi) = \LCSt$, thus by Lemma~\ref{lem:random-matrix} we have that  $\len(\pi^*) = \LCSt$ with probability at least $9/10$. The desired result follows from the observation that $\len(\pi^*)=|C_{\pi^*}|=n-\HAMt$.
\end{proof}


\section{A lower bound for the information transfer \label{sec:main}}

We are now able to prove Lemma~\ref{lem:encoding} which gives us lower bounds for the expected size of the information transfer of a node. Our approach extends that of \Patrascu and Demaine from~\cite{PD2006:Low-Bounds}.  Let $v$ be a
node of the information transfer tree $\calT$ and recall that the expected length of any encoding of the outputs $\Yv$ is an upper bound on its entropy.  \Patrascu and Demaine showed that it is possible to bound the conditional entropy of $Y_v$ in terms of the expected information transfer size $\mathbb{E}[I_v]$ by using what we will call an \emph{address-based} encoding scheme.  A shortcoming of this encoding, which we will have to overcome, is that storing the address of a cell could require $\Theta(\log n)$
bits, making the length of the encoding too large to be useful for small $w$, that is $w\in o(\log n)$.  In the following lemma we have slightly generalised the original statement of this bound from~\cite{PD2006:Low-Bounds} to make the role that the word size $w$ plays explicit.



\begin{lem}
[\Patrascu and Demaine~\cite{PD2006:Low-Bounds}]\label{lem:H-upper-old}
There is a positive constant $\alpha$ such that for any node $v$
of the information transfer tree $\calT$, the entropy 
$$
H(\Yv\mid\Xvknown=\Xvfix)~\leq~(w+\alpha\cdot\log n)\cdot\expecteddisplay{\vit\bigm|\Xvknown=\Xvfix}.
$$

\end{lem}
Towards a new encoding that circumvents the limitations of Lemma~\ref{lem:H-upper-old}
we propose, as an intermediate step, another encoding which is used
to prove the following lemma. We refer to this encoding as the \emph{counting-based}
encoding.
\begin{lem}
\label{lem:H-upper-new} For any node $v$ of the information transfer
tree $\calT$, the entropy 
$$
H(\Yv\mid\Xvknown=\Xvfix)~\leq~\expecteddisplay{\log{\vtime \choose \vit}+w\cdot\vit+\log\vit~\Bigm|~\Xvknown=\Xvfix}.
$$
\end{lem}
\begin{proof}
We keep a counter of the cell reads performed by any correct algorithm for our online problems. For
each cell $c\in\vitset$, we store the contents of~$c$ and the value
of the cell-read counter for when $c$ is read for the first time
during the time interval in which $\Yv$ is outputted. To see why
this encodes $\Yv$, under a fixed $\Xvknown$, we use the algorithm
as a decoder. That is, first run the algorithm on known inputs until
the first symbol in $\Xv$ arrives. Then skip over all inputs in $\Xv$
and start simulating the algorithm from the beginning of the interval
where $\Yv$ is outputted. For every cell~$c$ being read, use the
cell-read counter to determine whether~$c$ is in $\vitset$ or not.
If it is, read its contents from the encoding of $\vitset$, otherwise
we already have the correct value from running the algorithm on $\Xvknown$.

The cell-read counter and contents of the cells in $\vitset$ can
be encoded succinctly as follows. During the time interval in which
the algorithm outputs $\Yv$, there are ${\vtime \choose \vit}$ possible
scenarios of when the cells in $\vitset$ are read for the first time.
Thus, under a predefined enumeration of these scenarios, we can specify
in exactly $\log{\vtime \choose \vit}$ bits when each cell in $\vitset$
is read. To do so we would also need to encode the size of $\vitset$,
which can be done in $\log\vit$ bits. The contents of the cells in
$\vitset$ is encoded in a total of $w\cdot\vit$ bits and stored
sorted by the time of the cell read.  Observe that it suffices to
store only the first read of any cell in $\vitset$ as the decoder
remembers every cell it has already accessed.
\end{proof}

To obtain our new lower bound for $\mathbb{E}[I_v]$ we will combine these two encodings schemes.
More precisely, let $\vcap = \ell_v \cdot \delta \cdot \log^2 n$ define a threshold value. When $\vtime\geq \vcap$, we use the address-based encoding, and when $\vtime<\vcap$  we use the counting-based encoding.
The threshold has been chosen carefully so that we utilise the strengths of each of the two encodings.
Only one bit of additional information is required to specify which of the two encodings is used.

We now use the fact that Lemma~\ref{lem:encoding} only relates to fast nodes $v$ and use a combination of the two encodings and the fact that $v$ is either a high-entropy or a medium-entropy node. This gives us both upper and lower bounds on the conditional entropy of $\Yv$ which after some manipulation will provide us with our desired lower bounds for $\mathbb{E}[I_v]$.

\begin{proof}
[Proof of Lemma~\ref{lem:encoding}]We begin by proving the Hamming
distance and convolution part of the lemma when $w\geq\log n$. In
this case, the bound in the statement of the lemma simplifies to

\[
\expected{\vit}\in \Omega{\left(\frac{\delta\cdot\ell_{v}}{w}\right)}.
\]
Regardless of whether the node $v$ is fast or not, this lower bound
on $\expected{\vit}$ has already been established in previous work~\cite{CJ:2011,CJS:2013}
by using the address-based encoding. In the rest of the proof we will
therefore focus on the case where $w<\log n$. Here the concept
of a fast node will be important.

The encoding we use when $w<\log n$ is a combination of the address-based
encoding of Lemma~\ref{lem:H-upper-old} and the new counting-based
encoding of Lemma~\ref{lem:H-upper-new}. We refer to this combined
encoding as the \emph{mixed encoding}. Let 
\[
\vcap=\ell_{v}\cdot\delta\cdot\log^{2}n
\]
be a threshold value such that when $\vtime\geq\vcap$ we use the
address-based encoding, and when $\vtime<\vcap$ we use the counting-based
encoding. The threshold has been chosen carefully so that we utilise
the strengths of each one of the two encodings. One single bit of
information is sufficient to specify which encoding is being used.

Suppose $v$ is a high-entropy node when proving the Hamming distance
and convolution part of the lemma, and suppose $v$ is a medium-entropy
node when proving the edit distance and LCS part. The proof is almost
identical for both parts so we will combine them as follows. Let
\[
f(n)=\begin{cases}
\frac{1}{2}\cdot k & \mbox{when proving the Hamming distance and convolution part,}\\
\frac{1}{2}\cdot\frac{k}{\sqrt{\log n\cdot\log\log n}} & \mbox{when proving the edit distance and LCS part,}
\end{cases}
\]
where $k$ is the constant from either Definition~\ref{def:high-node}
of a high-entropy node or Definition~\ref{def:medium-node} of a
medium-entropy node, respectively. Thus,
\begin{equation}
H(\Yv\mid\Xvknown=\Xvfix)\,\geq\,2\cdot f(n)\cdot\delta\cdot\ell_{v},\label{eq:H-fn}
\end{equation}
where $\Xvfix$ is a fixed value. Observe the factor $1/2$ in $f(n)$.
It is included for convenience as will be clear shortly. Recall that
by Definition~\ref{def:medium-node} of a medium-entropy node~$v$,
Inequality~(\ref{eq:H-fn}) is only guaranteed to hold for half of
the fixed values $\Xvfix$.

Let the random variable $\venc$ denote the size of the mixed encoding
of $\Yv$. Since the expected size of any encoding is an upper bound
on the entropy,
\[
\expected{\venc\mid\Xvknown=\Xvfix}\,\geq\, H(\Yv\mid\Xvknown=\Xvfix).
\]
Since $\Xvknown$ is chosen uniformly from all possible instances in the hard distribution for all
our problems, taking expectation over $\Xvknown$, gives 
\begin{equation}
\expected{\venc}\geq f(n)\cdot\delta\cdot\ell_{v}.\label{eq:encoding-lower}
\end{equation}
Here the factor $1/2$ in the definition of $f(n)$ comes in handy
as Inequality~(\ref{eq:H-fn}) is only guaranteed to hold for half
of the values of~$\Xvknown$ in the definition of a medium-entropy
node~$v$.

The goal is to derive a lower bound on $\expected{\vit}$. We condition
on the running time $\vtime$: 
\begin{equation}
\expected{\venc}~=~\Pr[\vtime\geq\vcap]\cdot\expected{\venc\mid\vtime\geq\vcap}~+~\Pr[\vtime<\vcap]\cdot\expected{\venc\mid\vtime<\vcap}.\label{eq:slow-fast}
\end{equation}
We will upper bound $\expected{\venc}$ by upper bounding each term
on the right hand side separately, starting with the first term. By
the definition of a fast node, 
\[
\expected{\vtime}\,\leq\,\ell_{v}\cdot\delta\cdot\log n.
\]
Using Markov's inequality, it follows that 
\[
\Pr[\vtime\geq\vcap]\,\leq\,\frac{1}{\log n}.
\]
Whenever $\vtime\geq\vcap$, the address-based encoding is used, for
which 
\[
\venc\,\leq\,(\alpha+1)\cdot\log n\cdot\vit,
\]
where $\alpha$ is the constant of Lemma~\ref{lem:H-upper-old} and
$w\leq\log n$. From the last two inequalities we conclude that the
first term of Equation~(\ref{eq:slow-fast}) is upper bounded by
$(\alpha+1)\cdot\expected{\vit}$.

We now upper bound the second term of Equation~(\ref{eq:slow-fast}).
Trivially, $\Pr[\vtime<\vcap]\leq1$. When $\vtime<\vcap$, the counting-based
encoding of Lemma~\ref{lem:H-upper-new} is used, hence 
\[
\venc\,=\,\log{\vtime \choose \vit}+w\cdot\vit+\log\vit\,\leq\,2\left(\log{\vtime \choose \vit}+w\cdot\vit\right).
\]
Using the fact ${a \choose b}\leq(a\cdot e/b)^{b}$ and the conditioning
$\vtime<\vcap$, we have 
\begin{align*}
\venc & \,\leq\,2(\vit\log\vcap+\vit\log e-\vit\log\vit+w\cdot\vit)\\
 & \,=\,2\vit\log\vcap+\vit\cdot k_{1}w-2\vit\log\vit,
\end{align*}
where $k_{1}=2+2\log e$ is a constant. Thus, by linearity of expectation,
\[
\expected{\venc\mid\vtime<\vcap}\,\leq\,\expected{\vit}\cdot2\log\vcap+\expected{\vit}\cdot k_{1}w+\expected{-2\vit\cdot\log\vit},
\]
where the last term is concave and therefore, by Jensen's inequality,
upper bounded by 
\[
\expected{\vit}\cdot(-2\log\expected{\vit}),
\]
which gives us 
\[
\mathbb{E}[\venc\mid\vtime<\vcap]\,\leq\,\expected{\vit}\cdot\big(2\log\vcap+k_{1}w-2\log\expected{\vit}\big).
\]

Using the upper bounds on the terms of Equation~(\ref{eq:slow-fast})
that we just derived, combined with the lower bound on $\expected{\venc}$
of Inequality~(\ref{eq:encoding-lower}), as well as substituting
$\vcap$ with its value $\ell_{v}\cdot\delta\cdot\log^{2}n$, we have
\begin{align*}
(\alpha+1)\cdot\expected{\vit}+\expected{\vit}\cdot\big(2\log\ell_{v}+2\log\delta+4\log\log n+k_{1}w-2\log\expected{\vit}\big)\\
\geq~f(n)\cdot\delta\cdot\ell_{v}
\end{align*}
or equivalently, 
\begin{align}
\expected{\vit} & \,\geq\,\frac{f(n)\cdot\delta\cdot\ell_{v}}{(\alpha+1)+2\log\ell_{v}+2\log\delta+4\log\log n+k_{1}w-2\log\expected{\vit}}\notag\\
 & \,\geq\,\frac{\frac{1}{2}f(n)\cdot\delta\cdot\ell_{v}}{\log\ell_{v}+\log\delta+2\log\log n+k_{2}w-\log\expected{\vit}},\label{eq:before-trick}
\end{align}
where $k_{2}=(a+1)+k_{1}$ is a constant. In order to lower bound
$\expected{\vit}$ it might be tempting to set the $-\log\expected{\vit}$
term to zero in the denominator above. Unfortunately, this bound will
be too weak for our purposes. Instead we perform the following arithmetic
manoeuvre.

First observe that if $\expected{\vit}$ is so large that the denominator
of Inequality~(\ref{eq:before-trick}) is negative then the statement
of Lemma~\ref{lem:encoding} follows by a suitable choice of constants.
We therefore continue under the assumption
that the denominator is positive.

In the hard distribution for Hamming distance and convolution, observe
that by Lemma~\ref{lem:high-entropy}, the information transfer $\vitset$ must always contain at least
one cell. Similarly by Lemma~\ref{lem:medium-entropy}, in the hard distribution for edit distance and LCS,
for at least half of the fixed values of $\Xvknown$ the information
transfer contains at least one cell. Thus, 
\[
\expected{\vit}\geq\frac{1}{2}.
\]
By replacing the $\log\expected{\vit}$ term in the denominator of
Inequality~(\ref{eq:before-trick}) with $-1$ we get a lower bound
on $\expected{\vit}$. We use $L$ as a shorthand for this lower bound,
hence 
\[
\expected{\vit}\,\geq\, L\,=\,\frac{\frac{1}{2}f(n)\cdot\delta\cdot\ell_{v}}{\log\ell_{v}+\log\delta+2\log\log n+k_{3}w}\,,
\]
and $k_{3}=k_{2}+1$, implying $k_{3}w\geq k_{2}w+1\geq k_{2}w-\log\expected{\vit}$.

We can now boost the lower bound $L$ of $\expected{\vit}$ by reconsidering
Equation~(\ref{eq:before-trick}) and replacing the $-\log\expected{\vit}$
term in the denominator with $-\log L$. That is, 
\begin{equation}
\expected{\vit}\,\geq\,\frac{\frac{1}{2}f(n)\cdot\delta\cdot\ell_{v}}{\log\ell_{v}+\log\delta+2\log\log n+k_{2}w-\log L}\,.\label{eq:long}
\end{equation}
By substituting $L$ with its value above, the denominator in Inequality~(\ref{eq:long})
can be written as 
\begin{align*}
\textup{} & \log\ell_{v}+\log\delta+2\log\log n+k_{2}w-\log(\frac{1}{2}f(n)\cdot\delta\cdot\ell_{v})\\
 & \phantom{=\,}+\log(\log\ell_{v}+\log\delta+2\log\log n+k_{3}w)\\
 & \leq\log\ell_{v}+\log\delta+2\log\log n+k_{2}w-\log\frac{1}{2}-\log f(n)-\log\delta-\log\ell_{v}\\
 & \phantom{=\,}+\log(\log n+\log n+2\log n+k_{3}\log n)\tag{\textup{recall \ensuremath{w\leq\log n}}}\\
 & \leq3\log\log n+k_{4}w-\log f(n)\\
 & \leq4\log\log n+k_{5}w\,,
\end{align*}
where $k_{4}$ and $k_{5}$ are constants, and where the last inequality
holds for either value of $f(n)$ above. We have also assumed that
$\delta\leq n$, that is we ignore the unnatural case where inputs
are exponential in $n$. Thus, by substituting the upper bound on
the denominator into Equation~(\ref{eq:long}), we have 
\[
\expected{\vit}\,\geq\,\frac{f(n)\cdot k_{6}\cdot\delta\cdot\ell_{v}}{w+\log\log n}\,,
\]
where $k_{6}$ is yet another constant. This inequality concludes the argument we need to prove Lemma~\ref{lem:encoding}. For Hamming distance and convolution the result now follows immediately and for the edit
distance and LCS we simply set $\delta=2$ and $w=1$.
\end{proof}

\section{A cell-probe algorithm for the edit distance problem}\label{sec:edit-upper}

In this section we prove Theorem~\ref{thm:edit-upper},
that is we show that there is a cell-probe algorithm that solves the online edit distance problem and runs in $O((\log^2 n)/w)$ amortised time per arriving symbol.
We want the algorithm to output
\begin{equation}
  d(i)=\min_{h\leq i}\,\ED{F,S[h,i]},\label{eq:minedit}
\end{equation}
where $F$ is the fixed string of length $n$, $S$ is the stream, and
$\ED{F,S[h,i]}$ is the smallest number of edit operations (i.e.,~replace,
delete and insert) required to transform $F$ into the substring $S[h,i]$. Recall
that at arrival $i$, only the symbols in $S[0,i]$ are known to the
algorithm. That is, the algorithm cannot know what symbols will arrive
in the future.

\subsection{Three assumptions}

We make three assumptions about the input in order to make the presentation
of our algorithm cleaner. The first assumption is about the alphabet
size. The symbols in both $F$ and $S$ are from the alphabet $[2^{\delta}-1]$
and we will assume that $2^{\delta}\leq n+1$, hence $\delta\in O(\log n)$.
As there can be at most $n$ distinct symbols in $F$, if the alphabet
contains more than $n+1$ symbols we can map every symbol not in $F$
to one specific symbol that does not occur in $F$. This mapping can easily
be done in $O(\delta)$ time, or $O(\log n$) time, if the alphabet
size is polynomial in $n$.

Secondly we assume that $n$ is a power-of-two. It is easy to extend
our algorithm to allow any $n$. Namely, pad $F$ at the left-hand
end with a new symbol $\sigma$, where $\sigma$ does not occur in
$S$, such that the length increases from $n$ to $n'$, where $n'$
is the smallest power-of-two greater than $n$. Call this new string
$F'$. Adding the symbol $\sigma$ increases the alphabet size by
one, hence increases $\delta$ by at most one. Now observe that 
\[
\text{Edit}(F,S[h,i])=\text{Edit}(F',S[h,i])-(n'-n).
\]
Thus, we solve the online edit distance problem for $F'$ and simply
subtract $n'-n$ from every output.

Lastly we will assume that the fixed array $F$ is known to the algorithm,
that is, for an arbitrarily given $F$ of length $n$ we show that
there is an algorithm that solves the edit distance problem such that
the algorithm performs on average $O((\log^{2}n)/w)$ cell probes per
symbol arrival. By ``hard-coding'' $F$ into the algorithm we avoid
charging for cell probes when accessing elements of $F$. In Section~\ref{sec:preprocessing}
we explain how this constraint can be lifted by adding a preprocessing
stage of $F$ before the symbols start arriving in the stream. Thus,
ultimately we do indeed give an algorithm that takes both $F$ and
the stream $S$ as input. The preprocessing uses a super-polynomial
number of cell probes.

\subsection{Dynamic programming and the underlying DAG}

The offline edit distance problem is traditionally solved with dynamic programming by
filling in a two-dimensional programming table. The dynamic programming
recurrence specifies a directed acyclic graph (DAG) such that the
optimal sequence of edit operations can be obtained by tracing the
edges backwards. We will simply refer to this graph as the \emph{DAG},
where the nodes form a two-dimensional lattice. The nodes are labelled
$(j,i)$, where $j\in\{-1\}\cup[n]$ is the row and $i\in\{-1\}\cup\mathbb{N}$
is the column. The rows are associated with the fixed string $F$
such that row $j$ is associated with $F[j]$. Row $-1$ is not associated
with any symbol of $F$ but is there two allow the empty string. Similarly,
the columns are associated with the stream $S$ such that column $i$
is associated with $S[i]$. Column $-1$ is not associated with any
symbol of $S$. Observe that the width of the DAG is unlimited and
grows as new symbols arrive.

The edges of the DAG have weights from the set $\{0,1\}$ to reflect
the dynamic programming recurrence. Node $(j,i)$ has the following
three edges leaving it:
\begin{lyxlist}{00.00.0000}
\item [{$\qquad\rightarrow$}] Edge $((j,i),(j,i+1))$. Weight is always~1
except when $j=-1$ in which case the weight is~0.
\item [{$\qquad\downarrow$}] Edge $((j,i),(j+1,i))$. Weight is always~1.
\item [{$\qquad\searrow$}] Edge $((j,i),(j+1,i+1))$. Weight is~0 if
$F[j+1]=S[i+1]$ and~0 otherwise.
\end{lyxlist}
We define $\calD(j,i)$ to be the weight of the smallest-weight path
from $(-1,-1)$ to $(j,i)$ in the DAG. It follows that
\[
d(i)=\calD(n-1,i),
\]
where $d(i)$, defined in Equation~(\ref{eq:minedit}), is the output
after the $i$-th arrival in the stream.

We will use the variable name $j$ exclusively to point to elements
of~$F$ and wherever we write \emph{for~all~$j$} we mean for all
$j\in\{-1\}\cup[n]$. The variable name $i$ will be used to point
to symbols of the stream~$S$.

\subsection{An algorithm for online edit distance}

Towards proving Theorem~\ref{thm:edit-upper} we first present an
unoptimised algorithm that solves the online edit distance problem. In Section~\ref{sec:upper-faster} we describe
an important speedup that enables us to give the final time complexity of $O((\log^2{n})/w)$ time per arriving symbol.

Our first algorithm for computing edit distance online is given in
Algorithm~\ref{alg:upper-first}, and will be explained below.  The algorithm fills in values from a dynamic programming table, denoted $D$, by reading in values from previous columns. In this respect the overall structure is similar to the classic offline dynamic programming solution for edit distance. In the standard dynamic programming solution new values are computed using values from the immediately preceding column, which is equivalent to setting $\pred(i)=i-1$ in Algorithm~\ref{alg:upper-first}.   Unfortunately, doing this would lead to a complexity of
$\Omega(n/w)$ per arriving symbol. To achieve a much faster solution we will need to choose a more suitable column, modify the definition of the classic dynamic programming table and then show how its values can be succinctly encoded to allow them to be read and written efficiently.

\begin{algorithm}[t]
\protect\caption{~~Online edit-distance in the cell-probe model (unoptimised version) \label{alg:upper-first} }

\vspace{8pt}


When a new symbol $S[i]$ arrives,%
{} 
\begin{itemize}
\item if $i=0$, compute $D(j,0)$ naively for all $j$ and output $D(n-1,0)$. 
\item if $i>0$, do the following.%

\begin{enumerate}
\item Read in the substring $S[\pred(i),i-1]$.\item[~]\emph{We now know
the part of the DAG spanned by the columns $\pred(i)$ through $i$
as the whole of $F$ is always known to the algorithm.}
\item Read in $D(j,\pred(i))$ for all $j$. \emph{These values are stored
as blocks.} 
\item Compute $D(j,i)$ for all $j$ and write these values to memory as blocks.
\item Output $D(n-1,i)$. \end{enumerate}
\end{itemize}
\end{algorithm}

For a positive integer $i$, let $\mask(i)$ be
the number obtained by taking the binary representation of $i$ and
setting the least significant~1 to~0. For example, if $i=212,$ that is
$11010100$ in binary, then $\mask(i)=208$ as this is $11010000$
in binary. Let
\[
\pred(i)=
\begin{cases}
  \mask(i) &\textup{if $\mask(i)>i-n$,}\\
  i-n &\textup{otherwise.}
\end{cases}
\]
When computing entries
of the $i$-th column of the table, values from column $\pred(i)$
will be used. The function $\pred(i)$ specifies a column
in the past and its definition ensures that we never look more
than $n$ columns back. We also define $\pred_{1}(i)=\pred(i)$ and for $k\geq2$, $\pred_{k}(i)=\pred(\pred_{k-1}(i))$. These iterated definitions will be useful for the analysis of our algorithm.

Let $\left\Vert (j,i)\rightsquigarrow(j',i')\right\Vert $ denote
the smallest weight of all paths in the DAG that go from node $(j,i)$
to~$(j',i')$, and if no such path exists, we define $\left\Vert (j,i)\rightsquigarrow(j',i')\right\Vert =\infty$.

We define the value $D(j,i)$ recursively. We first give the base case,
then introduce some supporting notation, and finally give the recursive
part of the definition.
\begin{defn}
[Base case]\label{def:D-base-case}$D(j,0)=\calD(j,0)$ for all $j$.
\end{defn}
In order to define $D(j,i)$ for $i>0$, as an intermediate step we
first define 
\begin{equation}
\Dinter(j,i)=\min_{j'\in\{-1,\dots,n-1\}}\Big(D\big(j',\pred(i)\big)+\left\Vert (j',\pred(i))\rightsquigarrow(j,i)\right\Vert \Big)\label{eq:Dinter}
\end{equation}
and observe that that $\Dinter(j,i)$ is not always finite.

We say that a sequence of positive integers in strictly decreasing order, where the difference between two consecutive numbers is 1, is a \emph{block}. Any sequence of positive integers can be \emph{decomposed} into consecutive blocks, for example, decomposing $9,8,7,7,6,5,4,3,4,3,2$ requires 3 blocks.   The first block is $9,8,7$, the second is $7,6,5,4,3$ and the third is $4,3,2$.

For $j=n-1$ down to~$-1$, the sequence of finite values $\Dinter(j,i)$ can be encoded greedily as a sequence of blocks. For example,
suppose that
\begin{align*}
\Dinter(n-1,i) & ,\dots,\Dinter(-1,i)=\\
 & \underbrace{18,17,16,15}_{1},\underbrace{15,14,13,12,11,10,9}_{2},\underbrace{10,9,8,7}_{3},\underbrace{7,6,5,4}_{4},\infty,\dots,\infty.
\end{align*}
Here we have enumerated the blocks 1, 2, 3 and~4.
We let $\progression(j,i)$ denote the block number
that contains the value~$\Dinter(j,i)$. Thus, in the example above, $\progression(n-1,i)=1$, $\progression(n-11,i)=2$ and
$\progression(n-12,i)=3$. To avoid splitting into two cases, depending
on whether $\Dinter(j,i)$ is finite or infinite, we will refer to a sequence of $\infty$-values as a block
as well. Thus, in the example above, $\progression(j,i)=5$ for all
$j\leq n-20$.
\begin{defn}
[Recursive part]\label{def:D-recursive-part}For $i>0$,
\[
D(j,i)=\begin{cases}
\Dinter(j,i) & \mbox{if \ensuremath{\progression(j,i)\leq3(i-\pred(i))},}\\
\infty & \mbox{otherwise.}
\end{cases}
\]

\end{defn}
Thus, $D(j,i)$ is finite only if $\Dinter(j,i)$ is contained within
the first $3(i-\pred(i))$ blocks. For any~$i$, this allows us to store $D(j,i)$ for all $j$ in $O((i-\pred(i)) \cdot \log n)$ bits. 

Before we analyse Algorithm~\ref{alg:upper-first} we will take a look at Figure~\ref{fig:alg-overview} which illustrates some of the values $D(j,i)$ computed by the algorithm.
\begin{figure}[t]
\centering{}\includegraphics{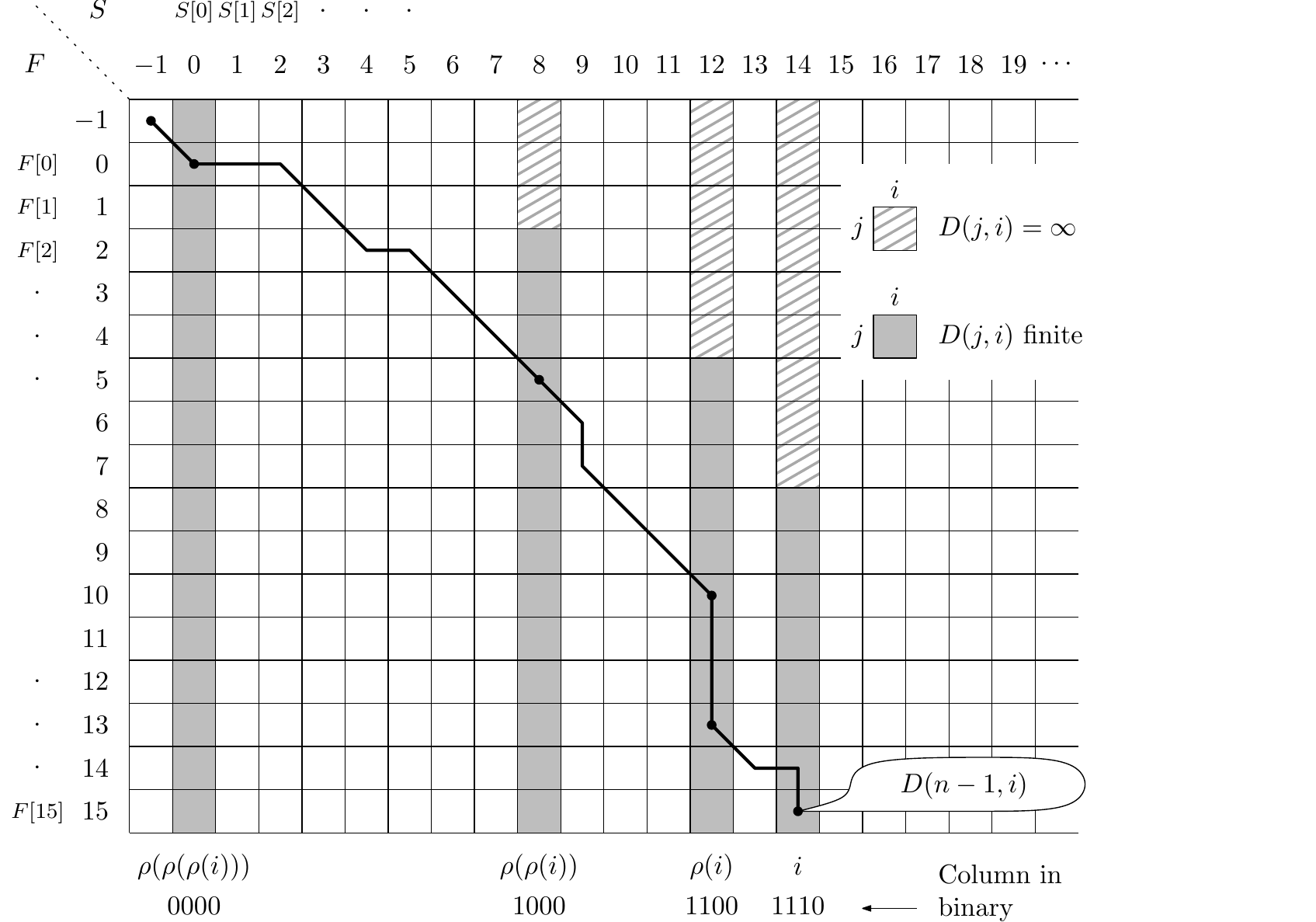} \protect\caption{An example of values $D(j,i)$ computed by Algorithm~\ref{alg:upper-first}. The path from $(-1,-1)$ to $(15,14)$ is a smallest-weight path. \label{fig:alg-overview}}
\end{figure}
Suppose that $i=14$ and symbol $S[i]$ has just arrived. Here $\pred(i)=12$, hence column~12 will be used in order to compute the values $D(j,14)$. Similarly, when computing the values $D(j,12)$, column $\pred(\pred(i))=8$ is used, and so on. By the definition of $D(j,i)$ it follows that all infinite values of a column must occur in sequence at the very top of the column.
The figure illustrates a path from node $(-1,-1)$ to $(n-1,i)$ that yields the value $D(n-1,i)$. As part of the correctness analysis we will demonstrate that this path always coincides with a smallest-weight path from $(-1,-1)$ to $(n-1,i)$ in the standard dynamic programming table that solves the edit distance problem. That is, despite setting some values of the dynamic programming table to infinity, there are always sufficiently many finite values left in order to correctly compute the output.

\subsection{Analysis of Algorithm~\ref{alg:upper-first} for online edit distance \label{sec:upper-first}}

To prove correctness of Algorithm~\ref{alg:upper-first} we need
to show that $D(n-1,i)=\calD(n-1,i)$ for all~$i$. We begin by proving
a useful, stronger fact that holds for certain values of~$i$.
\begin{lem}
\label{lem:complete-column}For any $i$ such that $D(j,i)$ is finite
for all $j$, $D(j,i)=\calD(j,i)$ for all $j$.\end{lem}
\begin{proof}
Let $i_{r}$ denote the $r$-th column $i$ for which $D(j,i)$ is
finite for all $j$. We use strong induction on $r$. For the base
case $r=1$, the column is~0 and the claim is true by Definition~\ref{def:D-base-case}.

For the induction step, suppose that the claim is true for all $i_{1},i_{2},\dots,i_{r}$
and suppose that $D(j,i_{r+1})$ is finite for all $j$. By Definition~\ref{def:D-recursive-part},
\[
D(j,i_{r+1})=\Dinter(j,i_{r+1})
\]
 for all~$j$. We will show that 
\begin{equation}
\Dinter(j,i_{r+1})=\calD(j,i_{r+1}).\label{eq:Dinter-complete}
\end{equation}

Observe that for any $i$ and $j$, if $D(j,i)$ is finite then so
is $D(j+1,i)$. Since $\Dinter(-1,i_{r+1})$ is finite, by Equation~(\ref{eq:Dinter}),
$D(-1,\pred(i_{r+1}))$ is finite, hence $D(j,\pred(i_{r+1}))$ is
finite for all~$j$. By the induction hypothesis, 
\[
D(j,\pred(i_{r+1}))=\calD(j,\pred(i_{r+1}))
\]
 for all~$j$. Thus, by Equation~(\ref{eq:Dinter}), Equation~(\ref{eq:Dinter-complete})
holds for all $ $$j$.
\end{proof}

Before showing that $D(n-1,i)=\calD(n-1,i)$ for all~$i$ we give
a property of smallest-weight paths in the DAG.
\begin{lem}
\label{lem:passes} For any $i\geq i'$ and $j\geq j'$, no smallest-weight
path from $(-1,-1)$ to $(j,i)$ can go via the node $(j',i')$ if
\begin{equation}
\calD(j',i')~>~\calD(j,i')-(j-j')+2(i-i').\label{eq:passes}
\end{equation}
\end{lem}
\begin{proof}
Let $P$ be any path in the DAG from $(-1,-1)$ to $(j,i)$ that passes
through $(j',i')$. The weight of $P$ is at least
\begin{equation}
\calD(j',i')+(j-j')-(i-i').\label{eq:smallest-path}
\end{equation}
To see this, first observe that the fact is immediately true if $(j-j')\leq(i-i')$.
For $(j-j')>(i-i')$, observe that any path from $(-1,-1)$ to $(j,i)$
via $(j',i')$ can contain at most $(i'-i)$ diagonal edges of weight
zero between $(j',i')$ and $(j,i)$.

Now suppose that Inequality~(\ref{eq:passes}) holds for $P$. We
will show that $P$ cannot be a smallest-weight path from $(-1,-1)$
to $(j,i)$. By combining Inequality~(\ref{eq:passes}) with~(\ref{eq:smallest-path})
we have that the weight of $P$ is strictly greater than
\[
\Big(\calD(j,i')-(j-j')+2(i-i')\Big)+(j-j')-(i-i')\,=\,\calD(j,i')+(i-i').
\]
To see why $P$ cannot be a smallest-weight path, consider the path
$P'$ that goes from $(-1,-1)$ to $(j,i)$ via the node $(j,i')$.
The weight of $P'$ is at most
\[
\calD(j,i')+(i-i')
\]
 as we can take a smallest-weight path from $(-1,-1)$ to $(j,i)$
and then follow $(i-i')$ horizontal edges to $(i,j)$. Thus, the
weight of~$P'$ is less than the weight of~$P$.
\end{proof}
We can now prove that the output from Algorithm~\ref{alg:upper-first} is correct.
\begin{lem}
\label{lem:cor} For all $i$, $D(n-1,i)=\calD(n-1,i)$.\end{lem}
\begin{proof}
Let $P$ be any smallest-weight path from $(-1,-1)$ to $(n-1,i)$
in the DAG. Figure~\ref{fig:correctness} illustrates an example
of $P$ and might be helpful when going through the proof. The proof
is by contradiction. Therefore, suppose that the Lemma is not true
for $i$.
\begin{figure}[t]
\centering{}\includegraphics{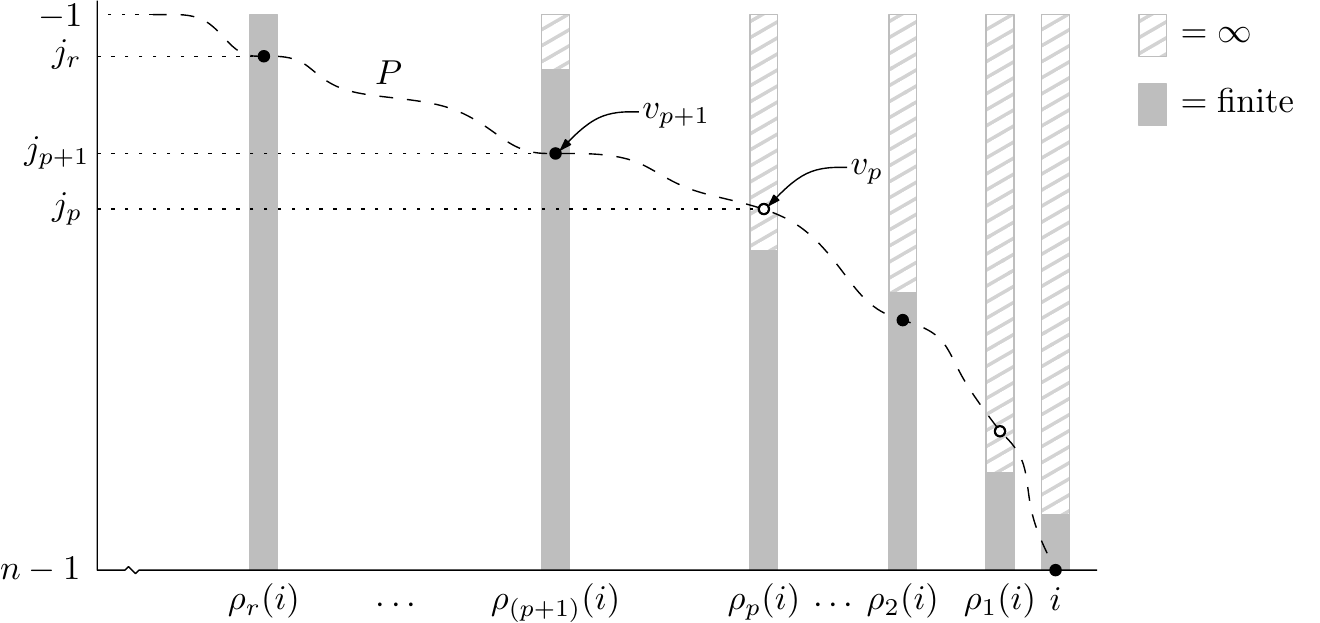} \protect\caption{Here $P$ is a smallest-weight path from $(-1,-1)$ to $(n-1,i)$.\label{fig:correctness}}
\end{figure}

By Lemma~\ref{lem:complete-column}, $D(j,i)$ cannot be finite for
all $j$, otherwise we have a contradiction. Let $r\geq1$ be the
smallest integer such that $D(j,\pred_{r}(i))$ is finite for all~$j$.
For $k\in\{0,\dots,r\}$, let $ $$(j_{k},\pred_{k}(i))$ be the last
node in column~$\pred_{k}(i)$ visited by $P$, where $j_{0}=n-1$
and $\pred_{0}(i)=i$. Let the node $v_{k}=(j_{k},\pred_{k}(i))$
so that we can write $D(v_{k})$ as a shorthand for $D(j_{k},\pred_{k}(i))$.
$ $ By Lemma~\ref{lem:complete-column}, 
\[
D(v_{r})=\calD(v_{r}).
\]
Let $p$ be the largest value in $\{0,\dots,r-1\}$ such that 
\begin{equation}
D(v_{p})\neq\calD(v_{p}).\label{eq:Dp}
\end{equation}
Since 
\[
D(v_{p+1})=\calD(v_{p+1}),
\]
and $P$ is a smallest-weight path, we have by the definition of $\Dinter(v_{p})$
in Equation~\ref{eq:Dinter} that 
\begin{equation}
\Dinter(v_{p})=\calD(v_{p}).\label{eq:Dtildep}
\end{equation}
Combining Equations~(\ref{eq:Dp}) and~(\ref{eq:Dtildep}) with
Definition~\ref{def:D-recursive-part} implies that 
\[
D(v_{p})=\infty.
\]
From Definition~\ref{def:D-recursive-part} it follows that
\[
3(\pred_{p}(i)-\pred_{p+1}(i))\,<\, n+1
\]
as the number of blocks per column can impossibly
exceed $n+1$. Thus, since $\pred_{p}(i)$ and $\pred_{p+1}(i)$ differ
by less than $(n+1)/3$, we have from the definitions of $\pred(i)$ and $\pred_{p}(i)$ that
\[
\pred_{p+1}(i)=\mask(\pred_{p}(i)),
\]
which implies that for all $k\in\{0,\dots,p\}$
\[
\pred_{k+1}(i)=\mask(\pred_{k}(i)).
\]
 In other words, $\pred_{k+1}(i)$ is
obtained from $\pred_{k}(i)$ by flipping the least significant~1
to~0 in the binary representation. We may therefore conclude that
\begin{equation}
\pred_{p}(i)-\pred_{p+1}(i)\,\geq\, i-\pred_{p}(i).\label{eq:pred-dist}
\end{equation}
To see why this is true, the following example might be helpful:
\[
\begin{aligned}i=\, & \texttt{{10010010100101100}}\\
\pred_{p}(i)=\, & \texttt{{10010010100000000}}\\
\pred_{p+1}(i)=\, & \texttt{{10010010000000000}}\\
i-\pred_{p}(i)=\, & \texttt{{00000000000101100}}\\
\pred_{p}(i)-\pred_{p+1}(i)=\, & \texttt{{00000000100000000}}
\end{aligned}
\]

We have assumed that the statement of the lemma is true, so in order
to show contradiction we will now argue that $v_{p}$ cannot be a
node on any smallest-weight path from $(-1,-1)$ to $(n-1,i)$, in
particular not on~$P$.

Since $\D(v_{p})\neq\Dinter(v_{p})$ we have from Definition~\ref{def:D-recursive-part}
that
\[
\progression(v_{p})\,>\,3(\pred_{p}(i)-\pred_{p+1}(i)).
\]
By the construction of the blocks it follows that
\[
\begin{aligned}\Dinter(n-1,\pred_{p}(i))-\Dinter(v_{p}) & \,\leq\,(n-1)-j_{p}-\big(\progression(v_{p})-1\big)\\
 & \,<\,(n-1)-j_{p}-\big(3(\pred_{p}(i)-\pred_{p+1}(i))-1\big)\\
 & \,\leq\,(n-1)-j_{p}-3(i-\pred_{p}(i))+1,
\end{aligned}
\]
where the last inequality follows uses Inequality~(\ref{eq:pred-dist}).
Observe that Equation~(\ref{eq:Dtildep}) implies that 
\[
\Dinter(n-1,\pred_{p}(i))=\calD(n-1,\pred_{p}(i)),
\]
which means that 
\[
\calD(n-1,\pred_{p}(i))-\calD(v_{p})\,<\,(n-1)-j_{p}-3(i-\pred_{p}(i))+1.
\]
Rearranging the terms gives
\[
\calD(v_{p})\,>\,\calD(n-1,\pred_{p}(i))-\big((n-1)-j_{p}\big)+3(i-\pred_{p}(i))-1,
\]
which can be written as
\begin{align*}
\calD(j',i') & \,>\,\calD(j,i')-(j-j')+3(i-i')-1\\
 & \,>\,\calD(j,i')-(j-j')+2(i-i'),
\end{align*}
where $j=n-1$, $j'=j_{p}$ and $i'=\pred_{p}(i)$. By Lemma~\ref{lem:passes}
we have the node $(j',i')$, or equivalently $v_{p}$, cannot
be a node on the smallest-weight path~$P$. The assumption that the
lemma is false is therefore not correct.
\end{proof}
We omit the analysis of the running time of Algorithm~\ref{alg:upper-first} as it is subsumed by that of Algorithm~\ref{alg:upper-second} which we now describe.

\subsection{A faster cell-probe algorithm for online edit distance \label{sec:upper-faster}}
We can speed up Algorithm~\ref{alg:upper-first} by modifying Step~2 as follows: instead of reading in
$D(j,\pred(i))$ for all $j$, only read in the values $D(j,\pred(i))$
covered by the first $8(i-\pred(i))$ blocks. This
modified version is given in Algorithm~\ref{alg:upper-second}. The
change has no impact on the correctness for the reason that any $j'$
in Equation~(\ref{eq:Dinter}) for which $\progression(j,\pred(i))\geq8(i-\pred(i))$
will never minimise $\Dinter(j,i)$. We prove this claim in Lemma~\ref{lem:progression-cap}
below, but first we give a supporting lemma.

\begin{algorithm}[t]
\protect\caption{~~Online edit-distance in the cell-probe model using $O((\log^{2}n)/w)$ probes \label{alg:upper-second} }

\vspace{8pt}

\emph{Time complexity $O((\log^{2}n)/w)$ per arriving symbol.}\vspace{8pt}

The algorithm is identical to Algorithm~\ref{alg:upper-first}, only
that Step~2 is replaced with this step:
\begin{enumerate}
\item [2.]Read in the values $D(j,\pred(i))$ that are covered by the first
$8(i-\pred(i))$ blocks. Any $D(j,\pred(i))$ not
covered is set to~$\infty$.\end{enumerate}
\end{algorithm}

\begin{lem}
\label{lem:one-diff} For any $i$ and $j$ such that $\Dinter(j-1,i)$
is finite, $\Dinter(j-1,i)\leq\Dinter(j,i)+1$.\end{lem}
\begin{proof}
The proof is by strong induction on $i$. The lemma is immediately
true for $i=0$. Now assume the lemma is true for all $i<i'$. Let
$j^{*}$ be the value of any $j'$ that minimises the expression of
Equation~(\ref{eq:Dinter}) and let $P$ a minimising path from $(j^{*},\pred(i))$
to $(i,j)$. We consider two cases.

\textbf{Case 1} ($j^{*}=j$). Here $P$ consists entirely of horizontal
edges. By Equation~(\ref{eq:Dinter}), an upper bound on $\Dinter(j-1,i)$
can be obtained by using only horizontal edges from $(j^{*}-1,\pred(i))$.
By the induction hypothesis, 
\[
\Dinter(j^{*}-1,\pred(i))\,\leq\,\Dinter(j^{*},\pred(i))+1,
\]
hence $\Dinter(j-1,i)$ is upper bounded by $\Dinter(j-1,i)+1$.

\textbf{Case 2} ($j^{*}<j$). By tracing the path $P$ backwards,
starting at the end node $(j,i)$, let $(j-1,i')$ be the first node
visited when moving up from row~$j$. An upper bound on $\Dinter(j-1,i)$
can be obtained by using $P$ until the node $(j-1,i')$, after which
only horizontal edges are followed. Thus, $\Dinter(j-1,i)$ is at
most $\Dinter(j,i)+1$, where the $+1$ term applies if the edge on
$P$ from $(j-1,i')$ is a diagonal edge with weight~0.
\end{proof}
We can now prove correctness of Algorithm~\ref{alg:upper-second},
falling back on the correctness of Algorithm~\ref{alg:upper-first}.
\begin{lem}
\label{lem:progression-cap}For any $i$, $j$ and $j'$ such that
\[
D(j,i)=D\big(j',\pred(i)\big)+\left\Vert (j',\pred(i))\rightsquigarrow(j,i)\right\Vert 
\]
is finite, 
\[
\progression(j',\pred(i))\leq8(i-\pred(i)).
\]
\end{lem}
\begin{proof}
Figure~\ref{fig:speedup} illustrates the variables introduced in
the proof. Let $\Delta=i-\pred(i)$ and let node $v=(n-1,i)$ and
$w=(n-1,\pred(i))$. Let~$x=(j,i)$ be the topmost node in column~$j$
such that $D(x)$ is finite. Let $y=(j',\pred(i))$ be any node such
that
\[
D(x)=D(y)+\left\Vert y\rightsquigarrow x\right\Vert .
\]
Thus,
\begin{figure}[t]
\centering{}\includegraphics{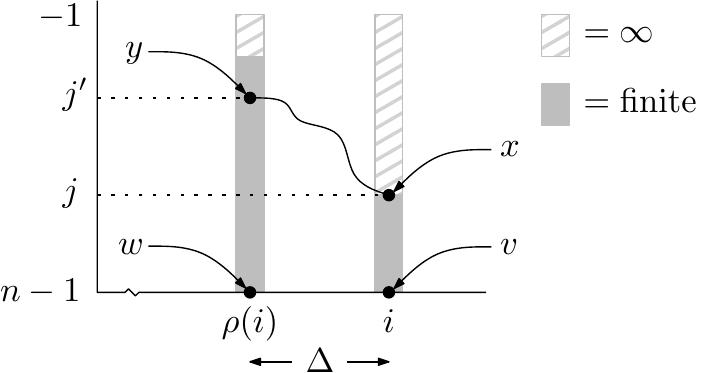} \protect\caption{Diagram supporting the proof of Lemma~\ref{lem:progression-cap}.\label{fig:speedup}}
\end{figure}
\begin{equation}
D(x)\,\geq\, D(y)+(j-j')-\Delta.\label{eq:speedup-contradiction}
\end{equation}

Now assume $\progression(y)>8\Delta$. We will show that this leads
to contradiction. First, observe that
\[
D(y)\,\geq\, D(w)-\big((n-1)-j'\big)+(\progression(y)-1)\,\geq\, D(w)-n+j'+8\Delta.
\]
By Definition~\ref{def:D-recursive-part}, $\progression(x)\leq3\Delta$.
By Lemma~\ref{lem:one-diff}, the start value of a block
is at most the end value of the previous block plus
one. Thus,
\[
D(x)\,\leq\, D(v)-\big((n-1)-j\big)+2(\progression(x)-1)\,\leq\, D(v)-n-1+j+6\Delta.
\]
Plugging the last two inequalities into Inequality~(\ref{eq:speedup-contradiction})
gives
\[
D(v)-n-1+j+6\Delta\,\geq\,\big(D(w)-n+j'+8\Delta\big)+(j-j')-\Delta
\]
which simplifies to
\[
D(v)\,\geq\, D(w)+\Delta+1.
\]
To see why this last inequality cannot be true, observe that $D(v)$
is never more than $D(w)+\Delta$, which is obtained through Equation~(\ref{eq:Dinter})
by using only horizontal edges from $w$ to~$v$. 
\end{proof}
It remains to argue that the running time of Algorithm~\ref{alg:upper-second}
is $O((\log^{2}n)/w)$ per arriving symbol. In Step~1 of the algorithm,
$i-\pred(i)$ symbols of $S$ are read. Each symbol is specified with
$\delta=O(\log n)$ bits. In Step~2, up to $8(i-\pred(i))$ blocks are read. Each one can be specified in $O(\log n)$ bits.
In Step~3, up to $3(i-\pred(i))$ blocks are written
to memory. Thus, when the symbol $S[i]$ arrives, no more than a constant
times $(i-\pred(i))\cdot\log n$ bits are read or written. To answer
the question of how many bits are read or written over a window of
$n$ arriving symbols, starting at any arrival~$i'$, we first give
the following fact.
\begin{lem}
\label{lem:second-bit-1}For any $i'>0$,
\[
\sum_{i=i'}^{i'+n-1}\big(i-\pred(i)\big)=O(n\log n).
\]
\end{lem}
\begin{proof}
Since we sum over $n$ consecutive values we may without loss of generality
assume that $i'=1$. Let $\first(i)$ denote the position of the least
significant~1 in the binary representation of $i$. For example,
if $i=110100$ (in binary), $\first(i)=2$. The sum can be written
as
\[
\sum_{i=1}^{n}2^{\first(i)}\,\leq\,\sum_{a=0}^{\log n}\big(2^{a}\cdot2^{(\log n)-a}\big)\,=\, O(n\log n).\tag*{\qedhere}
\]

\end{proof}
We conclude that over a window of $ $$n$ arriving symbols, a total
of $O(n\log n\cdot\log n)$ bits are read or written, hence the algorithm
performs $O(n(\log^{2}n)/w)$ cell probes. Amortised over the $n$
arriving symbols, the number of cell probes per arrival is 
\[
O{\left(\frac{\log^{2}n}{w}\right)}.
\]

\subsection{Preprocessing the fixed string $F$ \label{sec:preprocessing}}

Both Algorithms~\ref{alg:upper-first} and~\ref{alg:upper-second}
require that the fixed string~$F$ is known so that after Step~1,
the relevant part of the DAG can be determined. If $F$ had not been
known, the algorithm would have had to probe cells in order to also
read in a sufficiently large portion of $F$. Unlike the number $i-\pred(i)$
of symbols being read from $S$, the number of symbols needed from
$F$ could potentially span the whole string.

By exploiting the power of the cell-probe model, we may nevertheless
design a generic algorithm that takes $F$ as part of the input and
has a preprocessing step before the first symbol arrives in the stream.
Let 
\[
\Phi=\{0,\dots,n+1\}^{n+1}.
\]
For any~$i$, any sequence $D(-1,i),\dots,D(n-1,i)$ corresponds
to a unique element $\phi\in\Phi$ such that
\[
\phi=\big(D(-1,i),\dots,D(n-1,i)\big),
\]
where $D(j,i)=\infty$ is replaced with the value~$n+1$. When $F$
is part of the input, we may precompute all possible values written
to memory in Step~3 by considering $F$, every $\phi\in\Phi$ and every
string of maximum length~$n$ from
\[
\Gamma=\bigcup_{k=0}^{n}[2^{\delta}]^{k},
\]
where $[2^{\delta}]$ is the alphabet. The precomputed values are
inserted into a large dictionary. Thus, after Step~2, the values
to write to memory in Step~3 are fetched from the dictionary, where
the key is in $\Phi\times\Gamma$ and its value is in~$\Phi$. The
size of the dictionary is of course infeasibly large, and the preprocessing
stage involves an exponential number of cell writes. Nevertheless,
there is no additional cost of looking up a key in the dictionary
to the cost of probing the cells holding the key and the value associated
with it. 

The conclusion is that adding a preprocessing step and replacing Step~3
of Algorithm~\ref{alg:upper-second} with a dictionary lookup gives
us the desired upper bound of Theorem~\ref{thm:edit-upper}.

\printbibliography

%
%




\end{document}